\documentclass[svgnames,dvipsnames]{llncs}
\pdfoutput=1
\usepackage{xcolor}
\newcommand{\todo}[1][]{%
  \ifx/#1/%
    \textcolor{red}{TODO!}%
  \else%
    \textcolor{red}{todo: #1}%
  \fi%
}
\usepackage[hidelinks]{hyperref}
\usepackage{microtype}
\usepackage{amssymb}
\usepackage{enumerate}
\usepackage{amsmath}
\usepackage{xspace}
\usepackage{graphicx}

\spnewtheorem{observation}{Observation}{\bfseries}{\itshape}
\spnewtheorem{myclaim}{Claim}{\bfseries}{\itshape}
\newcommand{\claimqed}{\hfill$\lrcorner$}
\newcommand{\qedhere}{\qed}
\newenvironment{claimproof}[1][]{%
\ifx&#1&%
  \begin{proof}\renewcommand{\qed}{\claimqed}%
\else%
  \begin{proof}[#1]\renewcommand{\qed}{\claimqed}%
\fi%
}{\end{proof}\renewcommand{\qed}{\plainqed}}

\newcommand{\ONCF}{\textsc{ONCF-Coloring}\xspace}
\newcommand{\CNCF}{\textsc{CNCF-Coloring}\xspace}
\newcommand{\containment}{\ensuremath{\mathsf{NP  \subseteq coNP/poly}}\xspace}
\newcommand{\red}{\ensuremath{\text{\emph{red}}}\xspace}
\newcommand{\blue}{\ensuremath{\text{\emph{blue}}}\xspace}

\newcommand{\Oh}{\ensuremath{\mathcal{O}}\xspace}

\newcommand{\ba}{{\bf a}}
\newcommand{\bb}{{\bf b}}
\newcommand{\bi}{{\bf i}}
\newcommand{\bj}{{\bf j}}
\newcommand{\bu}{{\bf u}}
\newcommand{\bv}{{\bf v}}
\newcommand{\bw}{{\bf w}}
\newcommand{\br}{{\bf r}}

\newcommand{\drootCSP}{\textsc{$d$-Polynomial root CSP}\xspace}
\newcommand{\rootCSP}[1]{\textsc{\ensuremath{#1}-Polynomial root CSP}\xspace}

\newcommand{\defproblem}[3]{
 \vspace{1mm}
\noindent\fbox{
 \begin{minipage}{0.96\textwidth}
 \begin{tabular*}{\textwidth}{@{\extracolsep{\fill}}lr} \textsc{#1} &  \\ \end{tabular*}
 {\bf{Input:}} #2 \\
 {\bf{Question:}} #3
 \end{minipage}
 }
 \vspace{1mm}
}

\title{Parameterized Complexity of Conflict-free Graph Coloring
\thanks{This research was done with support by the NWO Gravitation grant NETWORKS. The research was partially done when the first and second author were associated with Eindhoven University of Technology.}}

\titlerunning{Conflict-free Graph Coloring} 

\author{Hans L. Bodlaender\inst{1} \and Sudeshna Kolay \inst{2} \and Astrid Pieterse \inst{3}}

\institute{Utrecht University, Netherlands.
\email{h.l.bodlaender@uu.nl}
\and
{Ben Gurion University of Negev, Israel.
\email{sudeshna.kolay@gmail.com}
}
\and
{Eindhoven University of Technology, Netherlands.
\email{astridpieterse@outlook.com}}
}
\date{}
\begin{document}
\maketitle
 \begin{abstract}
Given a graph $G$, a $q$-open neighborhood conflict-free coloring or $q$-ONCF-coloring is a vertex coloring $c\colon V(G) \rightarrow \{1,2,\ldots,q\}$ such that for each vertex $v \in V(G)$ there is a vertex in $N(v)$ that is uniquely colored from the rest of the vertices in $N(v)$. When we replace $N(v)$ by the closed neighborhood $N[v]$, then we call such a coloring a $q$-closed neighborhood conflict-free coloring or simply $q$-CNCF-coloring. In this paper, we study the NP-hard decision questions of whether for a constant $q$ an input graph has a $q$-ONCF-coloring or a $q$-CNCF-coloring.
We will study these two problems in the parameterized setting.
First of all, we study running time bounds on FPT-algorithms for these problems, when parameterized by treewidth. We improve the existing upper bounds, and also provide lower bounds on the running time under ETH and SETH. 
Secondly, we study the kernelization complexity of both problems, using vertex cover as the parameter.  We show that  both $(q \geq 2)$-ONCF-coloring and $(q \geq 3)$-CNCF-coloring cannot have polynomial kernels when parameterized by the size of a vertex cover unless \containment. On the other hand, we obtain a
polynomial kernel for  $2$-CNCF-coloring parameterized by vertex cover.
We conclude the study with some combinatorial results. Denote $\chi_{ON}(G)$ and $\chi_{CN}(G)$ to be the minimum number of colors required to ONCF-color and CNCF-color $G$, respectively. Upper bounds on $\chi_{CN}(G)$ with respect to structural parameters like minimum vertex cover size, minimum feedback vertex set size and treewidth are known. To the best of our knowledge only an upper bound on $\chi_{ON}(G)$ with respect to minimum vertex cover size was known. We provide tight bounds for $\chi_{ON}(G)$ with respect to minimum vertex cover size. Also, we provide the first upper bounds on $\chi_{ON}(G)$ with respect to minimum feedback vertex set size and treewidth. 
 \keywords{Conflict-free coloring, kernelization, fixed-parameter tractability, combinatorial bounds}
 \end{abstract}
 \section{Introduction}

Often, in frequency allocation problems for cellular networks, it is important to allot a unique frequency for each client, so that at least one frequency is unaffected by cancellation. Such problems can be theoretically formulated as a coloring problem on a set system, better known as conflict-free coloring~\cite{even}. Formally, given a set system $\mathcal{H} = (U, \mathcal{F})$, a $q$-conflict-free coloring $c \colon U \rightarrow \{1,2,\ldots,q\}$ is a function where for each set $f \in \mathcal{F}$, there is an element $v \in f$ such that for all $w \neq v \in f$, $c(v) \neq c(w)$. In other words, each set $f$ has at least one element that is uniquely colored in the set. This variant of coloring has also been extensively studied for set systems induced by various geometric regions~\cite{ajwani2012conflict,har,Smorodinsky2013}.

A natural step to study most coloring problems is to study them in graphs.  Given a graph $G$, $V(G)$ denotes the set of $n$ vertices of $G$ while $E(G)$ denotes the set of $m$ edges in $G$. A $q$-coloring of $G$, for $q \in \mathbb{N}$ is a function $c \colon V(G) \rightarrow \{1,2,\ldots,q\}$. The most well-studied coloring problem on graphs is proper-coloring. A $q$-coloring $c$ is called a proper-coloring if for each edge $\{u,v\} \in E(G)$, $c(u) \neq c(v)$.
In this paper, we study two specialized variants of $q$-conflict-free coloring on graphs, known as $q$-ONCF-coloring and $q$-CNCF-coloring, which are defined as follows.

\begin{definition}\label{def:oncf}
 Given a graph $G$, a $q$-coloring $c \colon V(G) \rightarrow \{1,2,\ldots,q\}$ is called a $q$-ONCF-coloring, if for every vertex $v \in V(G)$, there is a vertex $u$ in the open neighborhood $N(v)$ such that $c(u) \neq c(w)$ for all $w\neq u \in N(v)$. In other words, every open neighborhood in $G$ has a uniquely colored vertex.
\end{definition}

\begin{definition}\label{def:cncf}
 Given a graph $G$, a $q$-coloring $c\colon V(G) \rightarrow \{1,2,\ldots,q\}$ is called a $q$-CNCF-coloring, if for for every vertex $v \in V(G)$, there is a vertex $u$ in the closed neighborhood $N[v]$ such that $c(u) \neq c(w)$ for all $w\neq u \in N[v]$. In other words, every closed neighborhood in $G$ has a uniquely colored vertex.
\end{definition}

Observe that by the above definitions, the $q$-ONCF-coloring (or $q$-CNCF-coloring) problem is a special case of the conflict-free coloring of set systems. Given a graph $G$, we can associate it with the set system $\mathcal{H} = (V(G),\mathcal{F})$, where $\mathcal{F}$ consists of  the sets given by open neighborhoods $N(v)$ (respectively, closed neighborhoods $N[v]$) for $v \in V(G)$. A $q$-ONCF-coloring (or $q$-CNCF-coloring) of $G$ then corresponds to a $q$-conflict-free coloring of the associated set system.


Notationally, let $\chi_{CF}(\mathcal{H})$ denote the minimum number of colors required for a conflict-free coloring of a set system $\mathcal{H}$. Similarly, we denote by $\chi_{ON}(G)$ and $\chi_{CN}(G)$ the minimum number of colors required for an ONCF-coloring and a CNCF-coloring of a graph $G$, respectively. The study of conflict-free coloring  was initially restricted to combinatorial studies. This was first explored in~\cite{even} and~\cite{Smoro2003}. Pach and Tardos~\cite{pach} gave an upper bound of $\Oh(\sqrt{m})$ on $\chi_{CF}(\mathcal{H})$ for a set system $\mathcal{H} = (U,\mathcal{F})$ when the size of $\mathcal{F}$ is $m$. In~\cite{pach}, it was also shown that for a graph $G$ with $n$ vertices $\chi_{CN}(G) = \Oh(\log^2 n)$. This bound was shown to be tight in~\cite{glebov2014conflict}. Similarly,~\cite{Cheilaris2009} showed that $\chi_{ON}(G) = \Theta(\sqrt{n})$.

However, computing $\chi_{ON}(G)$ or $\chi_{CN}(G)$ is NP-hard. This is because deciding whether a $2$-ONCF-coloring or a $2$-CNCF-coloring of $G$ exists is NP-hard~\cite{GarganoRR15}. This motivates the study of the following decision problems under the lens of parameterized complexity.

\defproblem{$q$-\ONCF}{A graph $G$.}{Is there a $q$-ONCF-coloring of $G$?}

\noindent The $q$-\CNCF problem is defined analogously.

 Note that because of the NP-hardness for $q$-\ONCF or $q$-\CNCF even when $q=2$, the two problems are para-NP-hard under the natural parameter $q$. Thus, the problems were studied under structural parameters. Gargano and Rescigno~\cite{GarganoRR15} showed that both $q$-\ONCF and $q$-\CNCF have FPT algorithms when parameterized by (i) the size of a vertex cover of the input graph $G$, (ii) and the neighborhood diversity of the input graph.
Gargano and Rescigno
also mention that due to Courcelle's theorem, for a non-negative constant $q$, the two decision problems are FPT with the treewidth of the input graph as the parameter.

\paragraph{Our Results and Contributions.}
In this paper, we extend the parameterized study of the above two problems with respect to structural parameters.
Our first objective is to provide both upper and lower bounds for FPT algorithms when using treewidth as the parameter (Section~\ref{sec:alg}). We show that both $q$-\ONCF and $q$-\CNCF parameterized by treewidth~$t$ can be solved in time $(2q^2)^tn^{\Oh(1)}$. On the other hand, for $q\geq 3$, both problems cannot be solved in time $(q-\epsilon)^tn^{\Oh(1)}$ under Strong Exponential Time Hypothesis  (SETH). For $q=2$, both problems cannot be solved in time $2^{o(t)}n^{\Oh(1)}$ under Exponential Time Hypothesis (ETH). 

We also study the polynomial kernelization question (Section~\ref{sec:kernel-lb}).  Observe that both $q$-\ONCF and $q$-\CNCF cannot have polynomial kernels under treewidth as the parameter, as there are straightforward \textsc{and}-cross-compositions from each problem to itself.\footnote{This is true for a number of graph problems when parameterized by treewidth. For more information, see \cite[Theorem 15.12]{book}  and the example given for \textsc{Treewidth} (parameterized by solution size)  in \cite[page 534]{book}.} Therefore, we will study the kernelization question by a larger parameter, namely the size of a vertex cover in the input graph. The kernelization complexity of the \textsc{$q$-Coloring} problem (asking for a proper-coloring of the input graph) is very well-studied for this parameter, the problem admits a kernel of size $\widetilde{\Oh}(k^{q-1})$ \cite{JansenP17GraphColoringKernel} which is known to be tight unless \containment \cite{JansenK13GraphColoring}. From this perspective however, $q$-\CNCF and $q$-\ONCF turn out to be much harder: $q$-\CNCF for $q \geq 3$ and $q$-\ONCF for $q \geq 2$ do not have polynomial kernels under the standard complexity assumptions, when parameterized by the size of a vertex cover.
Interestingly, $2$-\CNCF parameterized by vertex cover size \emph{does} have a polynomial kernel and we obtain an explicit polynomial compression for the problem. Although this does not lead to a polynomial kernel of reasonable size, we study a restricted version called $2$-\CNCF-\textsc{VC-Extension} (Section~\ref{sec:extension-main}) and show that this problem has a $\Oh(k^2\log k)$ kernel where $k$ is the vertex cover size. Therefore, $2$-\CNCF behaves significantly differently from the other problems.

Finally, we obtain a number of combinatorial results regarding ONCF-colorings of graphs.
Denote by $\chi(G)$ the minimum $q$ for which a $q$-proper-coloring for $G$ exists. While $\chi_{CN}(G) \leq \chi(G)$, the same upper bound does not hold for $\chi_{ON}(G)$~\cite{GarganoRR15}. For a graph $G$, let ${\sf vc}(G)$, ${\sf fvs}(G)$ and ${\sf tw}(G)$ denote the size of a minimum vertex cover, the size of a minimum feedback vertex set and the treewidth of $G$, respectively. From the known result that $\chi(G) \leq {\sf tw}(G)+1 \leq {\sf fvs}(G) +1 \leq {\sf vc}(G) +1$, we could immediately obtain the fact that the same behavior holds for $\chi_{CN}(G)$. However, to show that $\chi_{ON}(G)$ behaves similarly more work needs to be done. To the best of our knowledge no upper bounds on $\chi_{ON}(G)$ with respect to ${\sf fvs}(G)$ and ${\sf tw}(G)$ were known, while a loose upper bound was provided with respect to ${\sf vc}(G)$ in~\cite{GarganoRR15}. We give a tight upper bound on $\chi_{ON}(G)$ with respect to ${\sf vc}(G)$ and also provide the first upper bounds  on $\chi_{ON}(G)$ with respect to ${\sf fvs}(G)$ and ${\sf tw}(G)$ (Section~\ref{sec:cb}).

Our main contributions in this work are structural results for the conflict-free coloring problem, which we believe gives more insight into the decision problems on graphs.  Firstly, the gadgets we build for the ETH-based lower bounds
could be useful for future lower bounds, but are also useful to understand difficult examples for conflict-free coloring which have not been known in abundance so far. We are able to reuse these gadgets in the constructions needed to prove the kernelization lower bounds. 
Secondly, our combinatorial results also give constructible conflict-free colorings of graphs and therefore provide more insight into conflict-free colored graphs. Finally, the kernelization dichotomy we obtain for $q$-\ONCF and $q$-\CNCF under vertex cover size as a parameter is a very surprising one.

\section{Preliminaries}\label{sec:prelim}

For a positive integer $n$, we denote the set $\{1,2,\ldots,n\}$ in short with $[n]$. For a graph $G$, given a $q$-coloring $c\colon V(G) \rightarrow [q]$ and a subset $S \subseteq V(G)$, we denote by $c \vert_{S}$ the restriction of $c$ to the subset $S$. For a graph $G$ that is $q$-ONCF-colored by a coloring $c$, for a vertex $v \in V(G)$, suppose $w\in N(v)$ is such that $c(w) \neq c(w')$ for each $w' \neq w \in N(v)$; then $c(w)$ is referred to as the ONCF-color of $v$. Similarly, for a graph $G$ that is $q$-CNCF-colored by a coloring $c$, for a vertex $v \in V(G)$, a unique color in $N[v]$ is referred to as the CNCF-color of $v$.

An \emph{edge-star graph} is a generalization of a star graph where there is a central edge $\{u,v\}$ and all other vertices $w$ have $N(w) = \{u,v\}$. A triangle is an example of an edge-star graph.

\subsection{Tree decompositions and treewidth}
\label{ssec:tree-decompositions-treewidth}
We define treewidth and tree decompositions.
\begin{definition}[Tree Decomposition ~\cite{book}]
A tree decomposition of a (undirected or directed) graph $G$ is a tuple $\mathcal{T} = (T,\{X_\bu\}_{\bu \in V(T)})$, where $T$ is a
tree in which each vertex $\bu \in V(T)$ has an assigned set of vertices $X_\bu \subseteq V(G)$ (called a bag) such
that the following properties hold:
\begin{itemize}
\item $\bigcup_{\bu \in V(T)} X_\bu = V(G)$.
 \item For any $xy \in E(G)$, there exists a $\bu \in V(T)$ such that $x, y \in X_\bu$.
\item If $x \in X_\bu$ and $x \in X_\bv$, then $x \in X_\bw$ for all $\bw$ on the path from $\bu$ to $\bv$ in $T$.
\end{itemize}
In short, we denote $\mathcal{T} = (T,\{X_\bu\}_{\bu \in V(T)})$ as $T$.
\end{definition}

The \emph{treewidth} $tw_{\mathcal{T}}$ of a tree decomposition $\mathcal{T}$ is the size of the largest bag of $\mathcal{T}$ minus one. A graph may have several distinct tree decompositions. The treewidth $tw(G)$ of
a graph $G$ is defined as the minimum of treewidths over all possible tree decompositions of $G$. Note that for the tree $T$ of a  tree decomposition, we denote a vertex of $V(T)$ in bold font. If $T$ is rooted at a vertex $\br$, for a vertex $\bu \in V(T)$, $V_{\bu} = \bigcup_{\bv \in T'} X_{\bv}$, where $T'$ is the subtree rooted at $\bu$.

A tree decomposition  ${\mathcal{T}}=(T,\{X_\bu\}_{\bu \in V(T)}))$ is called a {\em nice tree decomposition} if $T$ is a tree rooted at some node $\br$ where $ X_{\br}=\emptyset$, each node of $T$ has at most two children, and each node is of one of the following kinds:
\begin{itemize}
\item {\bf Introduce node}: a node $\bu$ that has only one child $\bu'$ where $X_{\bu}\supset X_{\bu'}$ and  $|X_{\bu}|=|X_{\bu'}|+1$.

\item {\bf  Forget vertex node}: a node $\bu$ that has only one child $\bu'$  where $X_{\bu}\subset X_{\bu'}$ and  $|X_{\bu}|=|X_{\bu'}|-1$.

\item {\bf Join node}:  a node  $\bu$ with two children $\bu_{1}$ and $\bu_{2}$ such that $X_{\bu}=X_{\bu_{1}}=X_{\bu_{2}}$.

\item {\bf Leaf node}: a node $\bu$ that is a leaf of $T$, and $X_{\bu}=\emptyset$.
\end{itemize}

One can  show that  a tree decomposition of width $w$ can be transformed into a nice tree decomposition of the same width $w$ and  with $\Oh(w |V(G)|)$ nodes, see~e.g.~\cite{book}.

We modify the definition of a nice tree decomposition slightly by ensuring that no bag in the tree decomposition is empty. This can easily be done by adding an arbitrary vertex $z \in V(G)$ to all bags of the current nice tree decomposition. This will ensure the non-emptiness property. Note that our nice tree decomposition will have width $w+1$.

\subsection{Parameterized complexity}
Let $\Sigma$ be a finite alphabet. A parameterized problem $\mathcal{Q}$ is a subset of $\Sigma^* \times \mathbb{N}$.
\begin{definition}[Kernelization]
Let $\mathcal{Q},\mathcal{Q}'$ be two parameterized problems and let $h \colon \mathbb{N} \to \mathbb{N}$ be some computable function. A \emph{generalized kernel} from $\mathcal{Q}$ to $\mathcal{Q}'$ of size $h(k)$ is an  algorithm that given an instance $(x,k)\in\Sigma^*\times\mathbb{N}$, outputs $(x',k')\in\Sigma^*\times\mathbb{N}$ in time $\text{poly}(|x| + k)$ such that
(i) $(x,k) \in \mathcal{Q}$ if and only if $(x',k') \in \mathcal{Q}'$, and
(ii) $|x'| \leq h(k)$ and $k' \leq h(k)$. \\
The algorithm is a \emph{kernel} if $\mathcal{Q} = \mathcal{Q}'$. It is a \emph{polynomial (generalized) kernel} if $h(k)$ is  a polynomial in $k$.
\end{definition}

Next, we describe a few methods that can be used to rule out the existence of polynomial kernels. One such method is by a polynomial parameter transformation \cite{BodlaenderTY11PPT} from a problem that is known to not admit a polynomial kernel. We repeat the necessary information here for completeness.
\begin{definition}[{Polynomial parameter transformation \cite{BodlaenderTY11PPT}}]
Let $\mathcal{Q}$ and $\mathcal{Q}'$ be parameterized problems. A \emph{polynomial parameter transformation} from $\mathcal{Q}$ to $\mathcal{Q}'$ is an algorithm that takes an input $(x,k)$ and outputs $(x',k')$ such that the following hold.
\begin{itemize}
\item $(x,k) \in \mathcal{Q}$ if and only if $(x',k') \in \mathcal{Q}'$, and
\item $k'$ is bounded by a polynomial in $k$.
\end{itemize}
We denote this as $\mathcal{Q} \leq_{\text{ppt}} \mathcal{Q}'$.
\end{definition}

The following Theorem follows from~\cite[Prop. 2.16]{BodlaenderJK14} and shows how to obtain lower bounds using polynomial parameter transformations.
\begin{theorem}[{\cite{BodlaenderJK14}}]\label{thm:ppt-works}
Let $\mathcal{Q}$ and $\mathcal{Q}'$ be parameterized problems with $\mathcal{Q} \leq_{\text{ppt}} \mathcal{Q}'$. If $\mathcal{Q}'$ admits a  polynomial generalized kernel, then $\mathcal{Q}$ admits a polynomial generalized kernel.
\end{theorem}

Another way to rule out the existence of polynomial kernels is using the framework of cross-compositions \cite{BodlaenderJK14}. We start by providing the necessary definitions.

\begin{definition}[Polynomial equivalence relation {\cite{BodlaenderJK14}}]
\label{def:eqvr} An equivalence relation $\mathcal{R}$ on $\Sigma^*$ is called a \emph{polynomial equivalence relation} if the following two conditions hold:
\begin{itemize}
\item There is an algorithm that given two strings $x,y \in \Sigma^*$ decides whether $x$ and $y$ belong to the same equivalence class in time polynomial in $|x|+|y|$.
\item For any finite set $S\subseteq \Sigma^*$ the equivalence relation $\mathcal{R}$ partitions the elements of $S$ into a number of classes that is polynomially bounded in the size of the largest element of $S$.
\end{itemize}
\end{definition}

\begin{definition}[{Cross-composition \cite{BodlaenderJK14}}] \label{def:cross-composition}
Let $L \subseteq \Sigma^*$ be a language, let $\mathcal{R}$ be a polynomial equivalence relation on $\Sigma^*$, and let $\mathcal{Q} \subseteq \Sigma^* \times \mathbb{N}$ be a parameterized problem. An \emph{\textsc{or}-cross-composition} of $L$ into $\mathcal{Q}$ (with respect to $\mathcal{R}$) is an algorithm that, given $t$ instances $x_1,x_2,\ldots,x_t \in \Sigma^*$ of $L$ belonging to the same equivalence class of $\mathcal{R}$, takes time polynomial in $\sum_{i=1}^t |x_i|$ and outputs an instance $(y,k)\in \Sigma^* \times \mathbb{N}$ such that the following hold:
\begin{itemize}
\item The parameter value $k$ is polynomially bounded in $\max_{i=1}^t |x_i| + \log{t}$, and
\item The instance $(y,k)$ is a yes-instance for $\mathcal{Q}$ if and only if at least one instance $x_i$ is a yes-instance for $L$.
\end{itemize}
\end{definition}
The following theorem shows how cross-compositions are used to prove kernelization lower bounds.
\begin{theorem}[{\cite{BodlaenderJK14}}] \label{thm:cross-composition-implies-LB}
If an NP-hard language $L$ \textsc{or}-cross-composes into the parameterized problem $\mathcal{Q}$, then $\mathcal{Q}$ does not admit a (generalized) polynomial kernelization unless \containment.
\end{theorem}

\subsection{Fast Subset Convolution Computation.}
Given a universe $U$ with $n$ elements, the subset convolution of two functions $f,g : 2^U \rightarrow \mathbb{Z}$ is a function $(f*g): 2^U \rightarrow \mathbb{Z}$ such that for every $Y \subseteq U$, $(f*g)(Y) = \Sigma_{X\subseteq Y} f(X)g(Y-X)$. Equivalently, $(f*g)(Y) = \Sigma_{A\uplus B=Y} f(A)g(B)$.

\begin{proposition}[{\cite{fomin2010exact}}]\label{prop:fast_subset_convolution}
 For two functions $f,g : 2^{U} \rightarrow \mathbb{Z}$, given all the $2^n$ values of $f$ and $g$ in the input, all the $2^n$ values of the subset convolution $f*g$ can be computed in $\Oh(2^n\cdot n^3)$ arithmetic operations.
\end{proposition}

In fact, this result can be extended to subset convolution of functions that map to any ring, instead of $(\mathbb{Z},+,\times)$~\cite{book}.
Consider the set $\mathbb{Z} \cup \{-\infty\}$, with the added relation that $\forall z \in \mathbb{Z},\{-\infty\} < z$. The $max$ operator takes two elements from this set and outputs the maximum of the two elements. Notice that $\mathbb{Z} \cup \{-\infty\}$, along with $max$ as an additive operator and $+$ as a multiplicative operator, forms a semi-ring~\cite{book}. We will call this semi-ring the integer max-sum semi-ring.  The subset convolution of two functions $f,g : 2^U \rightarrow \mathbb{Z}\cup \{-\infty\}$, with $max$ and $+$ as the additive and multiplicative operators, becomes $(f*g)(Y) = max_{A\uplus B=Y} f(A)+g(B)$.

\begin{proposition}[{\cite{fomin2010exact}}]\label{prop:minsum_subset_convolution}
 Given two functions $f,g: 2^U \rightarrow \{-M, \ldots,M\}$, all the $2^n$ values of $f$ and $g$ in the input, and all the $2^n$ values of the subset convolution $(f*g)$ over the integer max-sum semiring can be computed in time $2^n n^{\Oh(1)} \cdot \Oh(M\log M \log \log M)$.
\end{proposition}

For more details about subset convolutions and fast calculations of subset convolutions, please refer to \cite{book,fomin2010exact}.

\section{Algorithmic results parameterized by treewidth}\label{sec:alg}
In this section, we state the algorithmic results obtained for the \ONCF{} and \CNCF{} problems parameterized by treewidth.  On the algorithmic side, we have the following theorem.

\begin{theorem}
\label{thm:on-alg}
$q$-\ONCF and $q$-\CNCF parameterized by treewidth $t$ admits a $(2q^2)^tn^{\mathcal{O}(1)}$ time algorithm.
\end{theorem}

We also obtain algorithmic lower bounds for the problems under standard assumptions.

\begin{theorem}\label{thm:alg-lb}
The following algorithmic lower bounds can be obtained:
\begin{enumerate}
\item For $q \geq 3$, $q$-\ONCF{} or $q$-\CNCF parameterized by treewidth $t$ cannot be solved in $(q - \varepsilon)^tn^{\Oh(1)}$ time, under SETH.
\item $2$-\ONCF{} or $2$-\CNCF parameterized by treewidth $t$ cannot be solved in $2^{o(t)}n^{\Oh(1)}$ time, under ETH.
\end{enumerate}
\end{theorem}

In the remainder of this section, we will prove the two theorems stated above.

\subsection{Algorithms}\label{secappen:alg}

In this section, we prove Theorem~\ref{thm:on-alg}. In the following Lemma, we describe an algorithm for $q$-\ONCF, parameterized by treewidth. The algorithm for $q$-\CNCF parameterized by treewidth is very similar and has the same running time.
\begin{lemma}\label{lem:on-alg}
$q$-\ONCF parameterized by treewidth $t$ admits a $(2q^2)^tn^{\mathcal{O}(1)}$ time algorithm.
\end{lemma}

\begin{proof}
 We assume that a nice tree decomposition ${\mathcal{T}}=(T,\{X_\bu\}_{\bu \in V(T)})$, rooted at a leaf $\br$, is given to us. Also, recall that no bag in empty, and that each leaf bag or the root bag has exactly one vertex in it. We proceed with the following treewidth dynamic programming. Given a bag $X_{\bi}$ corresponding to the vertex $\bi \in V(T)$, a state for the bag is a tuple $(\bi,c_1,c_2,f)$, where

 \begin{itemize}
  \item $\bi$ determines the bag,
  \item $c_1\colon X_{\bi} \rightarrow [q]$ is a vertex coloring of $X_{\bi}$. Intuitively, for a vertex $x \in X_{\bi}$, $c_1(x)$ is the color $x$ receives in the conflict-free coloring we are after.
  \item $c_2\colon X_{\bi} \rightarrow [q]$ is a color assignment to each vertex of $X_{\bi}$. For a vertex $x \in X_{\bi}$, $c_2(x)$ should be the color that occurs exactly once in the neighborhood of $x$.
  \item $f\colon X_{\bi}\rightarrow \{0,1\}$ is an indicator function for the vertices of $X_{\bi}$. The idea is that $f(x)$ indicates whether $x$ already has a neighbor of color $c_2(x)$ in the subtree rooted at \bi.
\end{itemize}

Let $\mathcal{S}_{\bi}$ be the set of all states associated with $X_{\bi}$.
 A function $\Gamma_\bi\colon\mathcal{S}_{\bi}\rightarrow \{0,1\}$ is defined as follows: For a state $s = (\bi,c_1,c_2,f)$, suppose there is a  vertex coloring $c\colon V_{\bi} \rightarrow [q]$ such that (i) its restriction to the vertices in $X_{\bi}$ is the coloring $c_1$, (ii) for each $v \in X_\bi$, the color $c_2(v)$ is used at most once in $N(v) \cap V_{\bi}$, (iii) for each $v \in X_\bi$, if there is a a vertex $w \in N(v) \cap V_\bi$ such that $c_1(w) = c_2(v)$ then $f(v)=1$ and otherwise $f(v)=0$, (iv) for any vertex $v \in V_{\bi} \setminus X_{\bi}$, $N(v) \subseteq V_{\bi}$ has a uniquely colored vertex under coloring $c$. Then $\Gamma_\bi(s) =  1$. Otherwise, $\Gamma_\bi(s) = 0$. In other words, $c$ is such that except for the vertices in $X_{\bi}$ the graph induced on $V_{\bi}$ is ONCF-colored and a state $s$ stores a snapshot of $c$ at the boundary $X_{\bi}$ of the graph seen so far.

 Our dynamic programming will calculate the function $\Gamma_\bi$ for each bag $\bi$. Note that for the root $X_\br = \{z\}$, if in $\mathcal{S}_\br$ there is a state $(\br,c_1,c_2,f)$ such that $f(z) = 1$ and $\Gamma_\br(s) = 1$, then the graph $G$ has a $q$-ONCF-coloring. We describe our dynamic programming in cases according to the types of nodes of the tree decomposition.

 \paragraph*{Leaf Node:}
 Let $X_\bi= \{z\}$ be a leaf node. Then,

 \[
\Gamma_\bi(\bi,c_1,c_2,f) =
\begin{cases}
1 &\text{ if } f(z)=0\\
0 & \text{ otherwise.}
\end{cases}
\]
This can be calculated in $2q^2$ time. For the correctness, note that the uniquely colored neighbor of $z$ cannot appear in the graph seen so far as a leaf node only contains $z$.

 \paragraph*{Forget Node:}
 Let $X_\bi$ be a forget node with its child being $X_\bj$. Also, let $X_\bj = X_\bi \cup \{v\}$. Consider a state $s=(\bi,c_1,c_2,f)\in \mathcal{S}_\bi$ and a state $s'=(\bj,d_1,d_2,g) \in \mathcal{S}_\bj$. We say that $s'$ is consistent with $s$, or $s' \leq_c s$ if (i) $d_1|_{X_\bi} = c_1, d_2|_{X_\bi} = c_2$, (ii) $g|_{X_\bi} = f$ and $g(v) = 1$. Then,

 \[
\Gamma_\bi(\bi,c_1,c_2,f) =
\begin{cases}
\Gamma_\bj(\bj,d_1,d_2,g) &\text{ if } (\bj,d_1,d_2,g)\leq_c (\bi,c_1,c_2,f)\\
0 & \text{ otherwise.}
\end{cases}
\]
This can be calculated in $(2q^2)^{|X_\bi|}$ time. To prove correctness, first suppose $\Gamma_\bi(\bi,c_1,c_2,f) = 1$ and let $c\colon V_{\bi} \rightarrow [q]$ be a coloring that is a witness to this. By definition of consistency, there is only one state $(\bj,d_1,d_2,g)$ such that $(\bj,d_1,d_2,g)\leq_c (\bi,c_1,c_2,f)$. Since $V_{\bi} = V_{\bj}$, the same coloring $c$ also witnesses the fact that $\Gamma_\bj(\bj,d_1,d_2,g)=1$. Conversely, suppose $\Gamma_\bj(\bj,d_1,d_2,g) = 1$ where $(\bj,d_1,d_2,g)$ is the unique state such that $(\bj,d_1,d_2,g)\leq_c (\bi,c_1,c_2,f)$. Let $c\colon V_{\bj} \rightarrow [q]$ be a coloring that is a witness to this. Since $V_{\bi} = V_{\bj}$ and by definition of consistency $g(v)=1$, the same coloring $c$ also witnesses the fact that $\Gamma_\bi(\bi,c_1,c_2,f)=1$. Thus, our recurrence correctly calculates $\Gamma_\bi(\bi,c_1,c_2,f)$.

 \paragraph*{Introduce Node:}
 Let $X_\bi$ be an introduce node with its child being $X_\bj$. Also, let $X_\bi = X_\bj \cup \{v\}$. Consider a state $s=(\bi,c_1,c_2,f)\in \mathcal{S}_\bi$ and a state $s'=(\bj,d_1,d_2,g) \in \mathcal{S}_\bj$. We say that $s'$ is consistent with $s$, or $s' \leq_c s$ if (i) $c_1|_{X_\bj} = d_1, c_2|_{X_\bj} = d_2$, (ii) if there is a $w \in N(v) \cap X_\bi$ such that $c_2(v) = c_1(w)$ then there is exactly one such $w$ and $f(v) = 1$, otherwise there is no such $w$ and $f(v)=0$, (iii) If there is a $w \in N(v) \cap X_\bi$ such that $c_1(v) = c_2(w)$ then $g(w) = 0$ and $f(w)=1$, (iv) for all other $u\in X_\bj$, $f(u) = g(u)$. Then,

 \[
\Gamma_\bi(\bi,c_1,c_2,f) =
\begin{cases}
\max_{s' \in \mathcal{S}_j} \Gamma_j(s') &\text{ such that } s'\leq_c (\bi,c_1,c_2,f)\\
0 & \text{ otherwise.}
\end{cases}
\]
This can be calculated in $(2q^2)^{|X_\bi|}$ time. To prove correctness, first suppose $\Gamma_\bi(\bi,c_1,c_2,f) = 1$ and let $c\colon V_{\bi} \rightarrow [q]$ be a coloring that is a witness to this. By definition of consistency, there is at least one state $s'$ such that $s'\leq_c (\bi,c_1,c_2,f)$ and the same coloring $c$ also witnesses the fact that $\Gamma_\bj(s')=1$. Conversely, suppose there is a state $s'$ such that $s'\leq_c (\bi,c_1,c_2,f)$ and $\Gamma_\bj(s') = 1$. Then by definition of consistency, $\Gamma_\bi(\bi,c_1,c_2,f)=1$. Thus, our recurrence correctly calculates $\Gamma_\bi(\bi,c_1,c_2,f)$.

 \paragraph*{Join Node:}
 Let $X_\bi$ be a join node with its children being $X_\ba$ and $X_\bb$. This means that $X_\bi = X_\ba = X_\bb$. Consider a state $s=(\bi,c_1,c_2,f)\in \mathcal{S}_\bi$, and states $s'=(\ba,d_1,d_2,g) \in \mathcal{S}_\ba, s''=(\bb,e_1,e_2,h) \in \mathcal{S}_\bb$. We say that $\{s',s''\}$ is consistent with $s$, or $\{s',s''\} \leq_c s$ if (i) $c_1 = d_1 = e_1, c_2 = d_2 = e_2$, (ii) if there is a $v \in X_\bi$ such that $g(v) = 1$ ($h(v)=1$) then $h(v) = 0$ ($g(v)=0$) and $f(v) = 1$, (iii) for all other $v \in X_\bi$, $f(v)=g(v)=h(v)=0$. Then,

 \[
\Gamma_\bi(s) =
\begin{cases}
\max_{s' \in \mathcal{S}_\ba, s'' \in \mathcal{S}_\bb} \Gamma_\ba(s')\cdot \Gamma_\bb(s'') &\text{ such that } \{s',s''\}\leq_c s\\
0 & \text{ otherwise.}
\end{cases}
\]

 As before, the correctness of the recurrence follows from the definition of consistency.

It is straightforward to calculate this in $(3q^2)^{|X_\bi|}$ time, we will further improve this to $(2q^2)^{|X_\bi|}$ time as follows.

Notice that if we fix $c_1$ and $c_2$, then $d_1,d_2,e_1,e_2$ get fixed for consistent states. Also, given $f$, $g$, and $h$, consider the vertices $B(f) = \{v\in X_\bi \mid f(v)=0\}$. Then for each vertex $v \in B(f)$, $g(v) = h(v)= 0$. Now consider $X = X_\bi \setminus B(f)$. For consistent states, the following relations hold: (i) $g^{-1}(1)\uplus h^{-1}(1) = X$, (ii) $g^{-1}(0)\setminus B(f) = h^{-1}(1)$ and $g^{-1}(1) = h^{-1}(0)\setminus B(f)$. Thus, if we are given the function $g$, we can completely determine $h$ when we are looking at consistent states. Now, fix a function $f$. We define functions $F_\ba, F_\bb\colon2^{X_i}\rightarrow [0,1]$ in the following way. For a subset $Z \subseteq X_\bi \setminus B(f)$, define a function $g\colon X_\bi \rightarrow \{0,1\}$ such that for any $v\in Z$, $g(v) = 1$ and $g(v) = 0$ otherwise. Now, define $F_\ba(Z) = \Gamma_\ba(\ba,c_1,c_2,g)$ and $F_\bb(Z) = \Gamma_\bb(\bb,c_1,c_2,g)$.

Then,
 \[
\Gamma_\bi(\bi,c_1,c_2,f) =
\begin{cases}
1  &\text{ if } {F_a*F_b(X_\bi \setminus B(f))} = \\&\text{\quad}\max_{Z \subseteq X_\bi \setminus B(f)} \{{F_a(Z)+F_b((X_\bi \setminus B(f))\setminus Z)} \} = 2,\\
0 & \text{ otherwise.}
\end{cases}
\]

The correctness of this recurrence is same as the correctness of the previous recurrence. Due to fast subset convolution over the max-sum semi-ring~\cite{book}, this can be calculated in $(2q^2)^{|X_\bi|}$ time.
\qed\end{proof}

\subsection{Running time lower bounds}\label{secappen:lb}
In this section, we given the proof of Theorem~\ref{thm:alg-lb} by describing lower bounds on algorithmic running times for the \ONCF{} and \CNCF{} problems parameterized by treewidth.

We start by providing a running time lower bound on $2$-\ONCF under ETH claimed in Theorem \ref{thm:alg-lb}. The bound will be obtained by giving a reduction from $3$-SAT, and in order to give the reduction we will need the following type of gadget.

\begin{definition}\label{def:oncf-gadget}
An \emph{ONCF-gadget} is a gadget on ten vertices, as depicted in Figure \ref{fig:ONCF-gadget}.
\end{definition}

\begin{figure}
\centering
\includegraphics{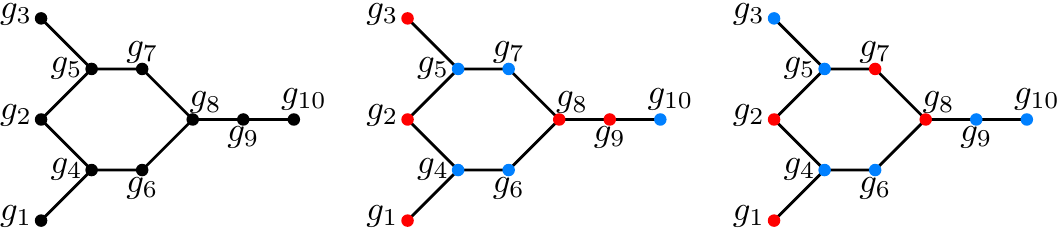}
\caption{The ONCF-gadget (left). Observe that if $g_1,g_2$, and $g_3$ are all red, then $g_9$ must also be red (middle), and if one of $g_1,g_2$, or $g_3$ is blue, then $g_9$ may be blue (right).}
\label{fig:ONCF-gadget}
\end{figure}

The objective of this gadget is the following. The vertices $\{g_1,g_2,g_3,g_{10}\}$ in Figure~\ref{fig:ONCF-gadget} will be the interaction points of the ONCF-gadget with the outside world. As will be proved in the following two lemmas, the gadget is designed so as to (i) disallow certain $2$-ONCF-colorings and (ii) allow certain $2$-ONCF-colorings on its interaction points.

\begin{lemma}\label{lem:ONCF-gadget:all-red}
Let $G$ be a ONCF-gadget with a coloring $c \colon V(G) \rightarrow \{\red,\blue\}$ such that for all $4 \leq i \leq 9$ the neighborhood of $g_i$ is ONCF-colored by $c$. If $c(g_1) = c(g_2) = c(g_3) = \red$, then $c(g_9) = \red$.
\end{lemma}

\begin{proof}
Suppose $c(g_1) = c(g_2) = c(g_3) = \red$. Since $N(g_4) = \{g_1,g_2,g_6\}$, this implies $c(g_6) = \blue$. Similarly, we find $c(g_7) = \blue$. Since $N(g_8) = \{g_6,g_7,g_9\}$ now has two \blue vertices, we conclude that $c(g_9) = \red$.\qed
\end{proof}

\begin{lemma}\label{lem:ONCF-gadget:not-all-red}
Let $G$ be a ONCF-gadget. Let $c' \colon \{g_1,g_2,g_3\} \rightarrow \{\red,\blue\}$ be a partial $2$-ONCF-coloring of $G$. If there exists $i \in [3]$ such that $c'(g_i) = \blue$, then $c'$ can be extended to a coloring $c$ satisfying
\begin{enumerate}
\item \label{prop:ONCF-colored} For every $4 \leq i \leq 9$, the neighborhood of vertex $g_i$ is ONCF-colored by $c$ (contains at most one red, or at most one blue vertex), and
\item \label{prop:fixed-colors} $c(g_9) = \blue$, $c(g_8) = \red$, $c(g_4) = c(g_5) = \blue$, and $c(g_{10}) = \blue$.
\end{enumerate}
\end{lemma}

\begin{proof}
Let $c$ equal $c'$ on vertex $g_1,g_2$ and $g_3$ and define $c(g_9) := \blue$, $c(g_8) := \red$, $c(g_4) := c(g_5) := \blue$, and $c(g_{10}) := \blue$. If $c'(g_1) = \blue$ or $c'(g_2) = \blue$, define $c(g_6) := \red$ else define $c(g_6):= \blue$. If $c'(g_2) = \blue$ or $c'(g_3) = \blue$, define $c(g_7):= \red$, otherwise let $c(g_7) := \blue$. This completes the definition of $c$.
It is easy to verify that both requirements are satisfied by this coloring, refer to Figure \ref{fig:ONCF-gadget} for an example coloring.\qed
\end{proof}

Now that we have introduced the necessary gadgets, we can prove the running time lower bound for $2$-\ONCF.
\begin{lemma}\label{lem:on2}
$2$-\ONCF{} parameterized by treewidth $t$ cannot be solved in $2^{o(t)}n^{\Oh(1)}$ time, under ETH.
\end{lemma}

\begin{proof}
We show this by giving a reduction from $3$-{\sc SAT}. Given an instance of $3$-{\sc SAT} with variables $x_1,\ldots,x_n$ and clauses $C_1,\ldots,C_m$, create a graph $G$ as follows. Start by creating palette vertices $R,R'$, and $B$, and edges $\{R,R'\}$ and $\{R',B\}$. For each variable $i\in[n]$, create vertices $u_i,v_i,w_i$ and add edges $\{u_i,v_i\}$ and $\{v_i,w_i\}$. For the remainder of the construction we will reuse the ONCF-gadget as defined in Definition \ref{def:oncf-gadget}. For each $j \in [m]$, add an ONCF-gadget $G_j$ and connect $g_{10}$ of this gadget to $R$. Add vertices $s^1_j,s^2_j$, and $s^3_j$ and connect $s^b_j$ to $g_b$ in $G_j$ for $b \in [3]$.
Let clause $C_j := (\ell_1,\ell_2,\ell_3)$. Now if $\ell_b = x_i$ for some $i \in [n], b \in [3]$, connect $s^b_j$ to $u_i$. Similarly, if $\ell_b = \neg x_i$, connect $s^b_j$ to $w_i$. This concludes the construction of $G$, it remains to show that $G$ is $2$-ONCF-colorable if and only if the formula was satisfiable.

Suppose the satisfiability instance has satisfying assignment $\tau \colon \{x_1,\ldots,x_n\} \rightarrow \{0,1\}$, we show how to color $G$. Let $c(R) :=  c(R'):=\red$, and $c(B)  :=\blue$. Let $c(v_i) := \blue$ for all $i\in [n]$ and define $c(s^b_j):=\red$ for all $b \in [3]$, $j \in [m]$. Finally, if $\tau(x_i) = 1$, let $c(u_i) := \red$ and $c(w_i) := \blue$. Otherwise, let $c(u_i) := \blue$ and $c(w_i) := \red$. For each gadget $G_m$, vertex $g_b$ for $b \in [3]$ has neighbor $s^b_j$. Let $v \in \{u_i,w_i \mid i \in [n]\}$ be the other neighbor of vertex $s^b_j$. Define $c(g_b)$ such that $c(g_b) \neq c(v)$. Since the formula was satisfied by $\tau$, for each $ j \in [m]$ there hereby exists $b \in [3]$ such that $c(g_b) = \blue$. We use Lemma \ref{lem:ONCF-gadget:not-all-red} to extend the partial coloring to color gadget $G_m$, with $c(g_{10}) = \blue$ and $c(g_4) = c(g_5) = \blue$. It is straightforward to verify that $c$ is a $2$-ONCF-coloring of $G$.

Suppose $G$ has a $2$-ONCF-coloring, we give a satisfying assignment $\tau$. Assume without loss of generality that $c(R) := \red$. Since $N(v_i) := \{w_i,u_i\}$ for all $i\in[n]$, it follows that $c(u_i) \neq c(w_i)$. We therefore define $\tau(x_i) := 1$ if $c(u_i) := \red$ and $\tau(x_i) = 0$ if $c(w_i) := \red$. Let $C_j$ be a clause, we will show that $\tau$ satisfies $C_j$ to conclude the proof. Suppose for contradiction that $\tau$ does not satisfy $C_j$. Then every vertex $s^b_j$ for $b \in [3]$ had one neighbor in $\{u_i,w_i \mid i \in [n]\}$ that is $\blue$ in $G$. Thereby, its only other neighbor $g_b$ in gadget $G_j$ must be colored \red. It follows from Lemma \ref{lem:ONCF-gadget:all-red} that $c(g_9) := \red$. Observe however that $N(g_{10}) := \{g_9,R\}$ and that both these vertices are \red, contradicting that $c$ is a $2$-ONCF-coloring of $G$. Thus, the formula is satisfied by $\tau$.

Note that the graph induced by $V(G)\setminus \{u_i,v_i,w_i\mid i \in [n]\}$ is a disjoint union of ONCF-gadgets and has treewidth two. As such, $G$ has treewidth at most $3n + 2$.

In this reduction a $3$-{\sc SAT} formula $\phi$ on $n$ variables and $m$ clauses is reduced to a graph $G$ with treewidth at most $3n+2$. We proved that $\phi$ is satisfiable if and only if $G$ has a $2$-ONCF-coloring. Since $3$-{\sc SAT} cannot be solved in $2^{o(n)}n^{\Oh(1)}$ time under ETH, this also implies that $2$-\CNCF{} parameterized by treewidth $t$ cannot be solved in $2^{o(t)}n^{\Oh(1)}$ time, under ETH.
\qed\end{proof}

Note that a reduction from $3$-{\sc SAT} to $2$-\ONCF{} was given in Theorem 2 of~\cite{GarganoRR15}. However, that reduction led to a quadratic blow-up in the input size. Hence, the need for the alternative reduction given above.

\begin{lemma}\label{lem:on3}
For $q \geq 3$, $q$-\ONCF{} parameterized by treewidth $t$ cannot be solved in $(q - \epsilon)^tn^{\Oh(1)}$ time, under SETH.
\end{lemma}

\begin{proof}
 It was shown in~\cite{SETH-tight} that for a constant $q \geq 3$, $q$-{\sc Coloring} cannot be solved in $(q - \epsilon)^t n^{\Oh(1)}$ time, under SETH. For a graph $G$, let $G'$ be the graph obtained by subdiving every edge of $E(G)$ once. It was shown in Theorem 3 of~\cite{GarganoRR15}, that $G$ has a $q$-coloring if and only if $G'$ has a $q$-ONCF-coloring. Also, note that $tw(G') \leq tw(G)$ since it is a subdivision of $G$. Thus, for a constant $q \geq 3$, the lower bound of $(q - \epsilon)^tn^{\Oh(1)}$ on the running time of any algorithm under SETH follows.
\qed\end{proof}

\begin{lemma}\label{lem:cn2}
$2$-\CNCF{} parameterized by treewidth $t$ cannot be solved in $2^{o(t)}n^{\Oh(1)}$ time, under ETH.
\end{lemma}

\begin{proof}
In~\cite{GarganoRR15}, a reduction of $2$-\CNCF{} was given from $3$-{\sc SAT}. In this reduction a $3$-{\sc SAT} formula $\phi$ on $n$ variables and $m$ clauses is reduced to a graph $G$ with treewidth at most $20m$. It was shown that $\phi$ is satisfiable if and only if $G$ has a $2$-CNCF-coloring. Since $3$-{\sc SAT} cannot be solved in $2^{o(m)}n^{\Oh(1)}$ time under ETH, this also implies that $2$-\CNCF{} parameterized by treewidth $t$ cannot be solved in $2^{o(t)}n^{\Oh(1)}$ time, under ETH.
\qed\end{proof}

\begin{lemma}\label{lem:cn3}
For $q \geq 3$, $q$-\CNCF{} parameterized by treewidth $t$ cannot be solved in $(q - \epsilon)^tn^{\Oh(1)}$ time, under SETH.
\end{lemma}

\begin{proof}
 It was shown in~\cite{SETH-tight} that for a constant $q \geq 3$, $q$-{\sc Coloring} cannot be solved in $(q - \epsilon)^t n^{\Oh(1)}$ time, under SETH. For a graph $G$, Theorem 3.1 of~\cite{AbelADFGHKS2017} constructs a graph $G'$ such that $G$ has a $q$-coloring if and only if $G'$ has a $q$-CNCF-coloring. 
  The construction of $G'$ requires the graphs $G_k$ as described in Section~\ref{subsec:cncf-col}, and first constructed in~\cite{AbelADFGHKS2017}. Recall that the $G_k$ is defined recursively as in Definition~\ref{def:G_k}.

 Returning to the construction of $G'$, we obtain $G'$ from $G$ in the following manner: (i) for each vertex $v \in V(G)$ we add two copies $G_q^{1v}$ and $G_q^{2v}$ of $G_q$ and make $v$ adjacent to all vertices of $G_q^{1v}$ and $G_q^{2v}$, (ii) for each edge $e=\{u,v\}\in E(G)$ we add two copies $G_{q-1}^{1e}$ and $G_{q-1}^{2e}$ of $G_{q-1}$ and make the vertices $u$ and $v$ adjacent to all vertices of $G_{q-1}^{1e}$ and $G_{q-1}^{2e}$. This completes the construction of $G'$. For the completion of our proof it remains to show that ${\sf tw}(G') \leq {\sf tw}(G)$ in order to obtain a lower bound of $(q - \epsilon)^tn^{\Oh(1)}$ on the running time of any algorithm under SETH.

 \begin{myclaim}\label{claim:G_k-tw}
 For a graph $G_k$, ${\sf tw}(G_k) \leq k-1$.
 \end{myclaim}
 \begin{claimproof}
 We prove our statement by induction on $k$. In the base case, it is true that ${\sf tw}(G_1) =0$ and ${\sf tw}(G_2)=1$. Let the induction hypothesis be that for any $k' < k$, ${\sf tw}(G_{k'}) \leq k'-1$. We prove the statement for $G_k$. By construction, $G_k$ contains a clique $C$ on $k$ vertices. We create a bag $X$ with all the vertices of $C$. By induction hypothesis, for each copy of $G_{k-1}$ we have a tree decomposition $\mathcal{T}_{k-1}$ of width $k-2$. Similarly, let $\mathcal{T}_{k-2}$ be a tree decomposition of $G_{k-2}$ with width $k-3$. Note that by construction, each copy of $G_{k-1}$ only has edges with a single vertex, say $v$ from the clique $C$. To each bag of the corresponding tree decomposition, we add the vertex $v$, thereby making the treewidth of the tree decomposition at most $k-1$. We pick an arbitrary bag of the tree decomposition and attach it to the bag $X$ containing the vertices of $C$. Similarly, each copy of $G_{k-2}$ only has edges with the end points of a single edge, say $\{u,v\}$ from the clique $C$. To each bag of the corresponding tree decomposition, we add the vertices $u,v$, thereby making the treewidth of the tree decomposition at most $k-1$. We pick an arbitrary bag of the tree decomposition and attach it to the bag $X$. The resulting tree decomposition has width at most $k-1$. Thus, ${\sf tw}(G_k) \leq k-1$.
\qed \end{claimproof}

 This helps us to show the desired treewidth bound for $G'$.

 \begin{myclaim}\label{claim:G'-tw}
 For a graph $G'$, ${\sf tw}(G') \leq \max\{{\sf tw}(G),q\}$.
 \end{myclaim}

 \begin{claimproof}
The construction of a desired tree decomposition is similar to the construction given in the previous Claim. Let $\mathcal{T}$ be a tree decomposition of $G$. By construction, each copy of $G_{q-1}$, that is added to $G$ to form $G'$, is attached to a single vertex in $V(G)$, say $v$. From the previous Claim, we have a tree decomposition $\mathcal{T}_{q-1}$ of width $q-2$ for this copy of $G_{q-1}$.  To each bag of $\mathcal{T}_{q-1}$, we add the vertex $v$, thereby increasing the treewidth to at most $k-1$. We pick an arbitrary bag of $\mathcal{T}$ that contains $v$ and an arbitrary bag of $\mathcal{T}_{q-1}$ and attach them together. Similarly, each copy of $G_{q-2}$, that is added to $G$ to form $G'$, is attached to the end points of a single edge in $E(G)$, say $\{u,v\}$. From the previous Claim, we have a tree decomposition $\mathcal{T}_{q-2}$ of width $q-3$ for this copy of $G_{q-2}$.  To each bag of $\mathcal{T}_{q-2}$, we add the vertices $u,v$, thereby increasing the treewidth to at most $k-1$. We pick an arbitrary bag of $\mathcal{T}$ that contains the edge $\{u,v\}$ and an arbitrary bag of $\mathcal{T}_{q-2}$ and attach them together. Note that the resulting tree decomposition is a tree decomposition of $G'$ and has width at most $\max\{{\sf tw}(G),q\}$. Thus, we are done. \qed\end{claimproof}

 Thus, ${\sf tw}(G') \leq \max\{{\sf tw}(G),q\} \leq {\sf tw}(G)$ since $q$ is a constant. Thus, for a constant $q \geq 3$, the lower bound of $(q - \epsilon)^tn^{\Oh(1)}$ on the running time of any algorithm under SETH follows.
\qed\end{proof}

Thus, using Lemmas~\ref{lem:on2}, \ref{lem:on3},  \ref{lem:cn3} and \ref{lem:cn2} we complete the proof of Theorem~\ref{thm:alg-lb}.

 \section{Kernelization}\label{sec:kernel-lb}
In this section, we will study the kernelizability of the ONCF- and CNCF-coloring problems, when parameterized by the size of a vertex cover. We prove the following two theorems to obtain a dichotomy on the kernelization question. 

\begin{theorem}\label{thm:noker}
$q$-\ONCF for $q\geq 2$ and $q$-\CNCF for $q \geq 3$, parameterized by vertex cover size do not have polynomial kernels, unless \containment.
\end{theorem}

Sections \ref{secappen:2-oncf-mp} and \ref{subsec:cncf-col} together give a full proof of  Theorem~\ref{thm:noker}.

\begin{theorem}\label{thm:genker}
$2$-\CNCF parameterized by vertex cover size $k$ has a generalized kernel  of size $\Oh(k^{10})$.
\end{theorem}

We prove the above theorem in Section \ref{secappen:poly-kernel-2CNCF}. Note that by using an NP-completeness reduction, this results in a polynomial kernel for $2$-\CNCF parameterized by vertex cover size. We also obtain an $\Oh(k^2\log k)$ kernel for an extension problem of $2$-\CNCF and this is described in Section \ref{sec:extension-main}.

\subsection{Kernel lower bounds for $q$-\ONCF}\label{secappen:2-oncf-mp}
In this part, we begin the proof of Theorem~\ref{thm:noker} by showing that $q$-\ONCF parameterized by vertex cover size has no polynomial kernel when $q$ is at least $2$. We first show the relevant bound for $q=2$ and then use a polynomial parameter transformation to obtain the general lower bound.

For the construction in the following proof, we will again use the ONCF-gadget that was introduced in Definition~\ref{def:oncf-gadget} (and shown in Figure \ref{fig:ONCF-gadget}). Recall the relevant properties of this gadget that were given in Lemmas~\ref{lem:ONCF-gadget:all-red} and~\ref{lem:ONCF-gadget:not-all-red}.

\begin{lemma}\label{lem:ONCF-kernel-LB}
$2$-\ONCF parameterized by vertex cover size does not have a polynomial kernel, unless \containment.
\end{lemma}
\begin{proof}
We show this by giving an \textsc{or}-cross-composition (see Definition \ref{def:cross-composition}) from \textsc{Clique} to $2$-\ONCF parameterized by vertex cover size. Note that \textsc{Clique} is an NP-hard problem~\cite{GareyJ90}. Therefore from Theorem~\ref{thm:cross-composition-implies-LB}, an \textsc{or}-cross-composition from \textsc{Clique} to $2$-\ONCF parameterized by vertex cover size implies that the latter does not have a polynomial kernel unless \containment.

We proceed with the description of the OR-cross-composition. For the sake of brevity, in this proof, we use $2$-ONCF-coloring and ONCF-coloring interchangeably. We define a polynomial equivalence relation $\mathcal{R}$ (see Definition \ref{def:eqvr}) as follows. Let 2 instances of \textsc{Clique} be equivalent under $\mathcal{R}$ if the graphs have the same number of vertices and they ask for a clique of the same size. It is easy to verify that $\mathcal{R}$ is a polynomial equivalence relation.  Suppose we are given $t$ instances of clique that are equivalent under $\mathcal{R}$, label them as $X_1,\ldots,X_t$. Let every instance have $n$ vertices and ask for a clique of size $k$, enumerate the vertices in each instance arbitrarily. We create an instance $G$ for $2$-ONCF-coloring by the following steps (see Figure \ref{fig:cross-composition-ONCF} for a sketch of $G$).
\begin{figure}[t]
\centering
\includegraphics{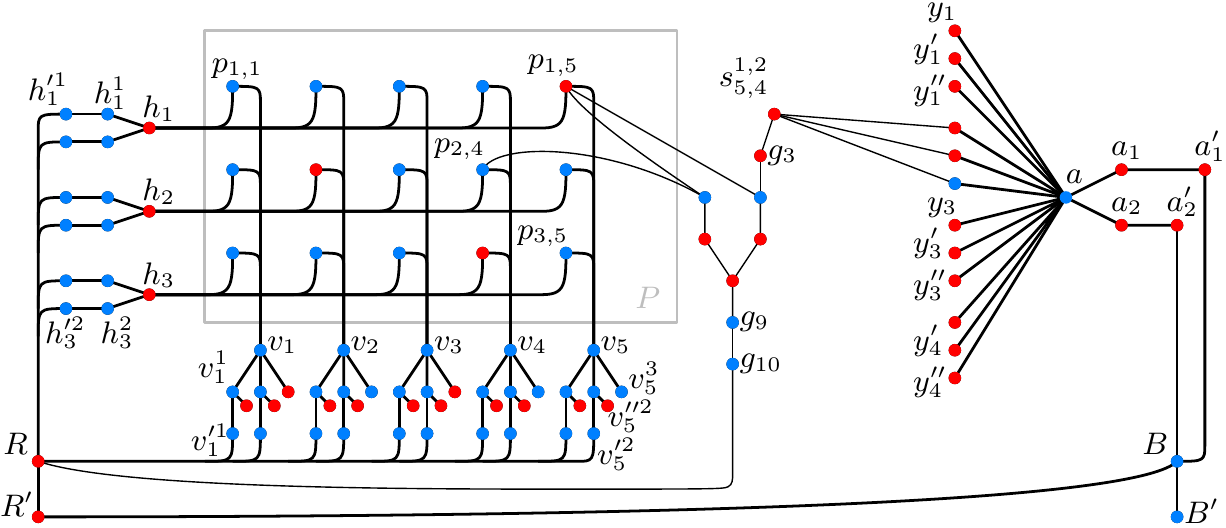}
\caption{A sketch of the constructed graph $G$  for $k = 3$, $n = 5$ and $t = 4$, assuming edge $\{5,4\}$ is missing in instance $X_2$ and present in all other instances. All vertices created in Steps \ref{step:ONCF:create-S} and \ref{step:ONCF:create-gadgets} of the construction are omitted for simplicity, except for vertex $s^{1,2}_{5,4}$ and gadget $G^{1,2}_{5,4}$.}
\label{fig:cross-composition-ONCF}
\end{figure}
\begin{enumerate}
\item \label{step:ONCF:create-RB} Create vertices $R$, $R'$, $B$, and $B'$. Connect $R$ to $R'$, $R'$ to $B$ and $B$ to $B'$. This ensures that $R$ and $B$ receive distinct colors in any coloring. (We will without loss of generality assume that $R$ receives the color \red, and $B$ receives the color \blue). Thus, the vertices $R$ and $B$ can be thought of as the palette for any $2$-ONCF-coloring for $G$.

\item \label{step:ONCF:create-Y} Create $3t$ vertices $\{y_\ell,y_\ell',y_\ell''\mid \ell \in [t]\}$ and let $Y$ be the set containing these vertices. These vertices will be used to ``select'' which instance has a clique of size $k$.
\item \label{step:ONCF:create-A} Add a vertex $a$ and connect $a$ to all vertices in $Y$. Add vertices $a_1$, $a_1'$, $a_2$, and $a_2'$ and edges $\{a,a_1\}$, $\{a,a_2\}$, $\{a_1.a_1'\}$, and $\{a_2,a_2'\}$. Finally, connect $a_1'$ to $B$ and connect $a_2'$ to $B$. This ensures that vertices $a_1$ and $a_2$ are red in any valid ONCF-coloring. Thereby, $a$ must have exactly one blue neighbor, implying exactly one vertex in $Y$ is blue. The vertex that is colored blue in $Y$ will then correspond to the input instance that has a clique of size $k$.
\item \label{step:ONCF:create-P} Add vertices $p_{i,j}$ for $i \in [k]$, $j \in [n]$. Let $P$ be the set consisting of these vertices. These vertices will be used to select the vertices that correspond to a clique in one of the input instances.
\item \label{step:ONCF:create-H} Add a vertex $h_i$ for all $i\in[k]$. Connect $h_i$ to $p_{i,j}$ for all $j\in[n]$. For each vertex $h_i$, add vertices $h^1_i,h'^1_i,h^2_i$, and $h'^2_i$. Connect $h$ to $h^1_i$ and $h^2_i$. Connect $h^1_i$ to $h'^1_i$, connect $h^2_i$ to $h'^2_i$. Connect $h'^1_i$ and $h'^2_i$ to $R$, in order to ensure that $h^1_i$ and $h^2_i$ will both be colored \blue in any ONCF-coloring. Let $H$ be the set of all vertices created in this step.
\item \label{step:ONCF:create-V} Add a vertex $v_j$ for $j \in [n]$. Connect $v_j$ to $p_{i,j}$ for all $i\in[k]$. Create vertex $v_j^3$ and connect it to $v_j$. Furthermore, create vertices $v_j^x$, $v'^x_j$, and $v''^x_j$ for all $x\in[2]$ and connect $v_j^x$ to $v_j$. Furthermore, connect $v''^x_j$ to $v^x_j$. Then, connect $v_j^x$ to $v'^x_j$ and $v'^x_j$ to $R$ for all $x\in[2]$. By this construction, vertex $v_j$ has at least two blue neighbors, and one neighbor whose coloring can be freely chosen. Let $V$ be the set of all vertices created in this step.
\item \label{step:ONCF:create-S} For every $i,i' \in [k]$ and $j,j' \in [n]$, we add a vertex $s^{i,i'}_{j,j'}$ and let $S$ be the set containing all these vertices. For each $\ell \in [t]$, $i,i' \in [k]$, and $j,j' \in [n]$, add the edges $\{s^{i,i'}_{j,j'}, y_\ell\}$, $\{s^{i,i'}_{j,j'}, y'_\ell\}$, and $\{s^{i,i'}_{j,j'}, y''_\ell\}$ if and only if $\{j,j'\} \notin E(X_\ell)$. The idea is that vertex $s^{i,i'}_{j,j'}$ verifies that if we select vertices $j$ and $j'$ to be part of the clique, then no instance $X_\ell$ where $\{j,j'\} \notin E(X_\ell)$ can be selected as the yes-instance. To do this, we add additional gadgets in the following step.
\item \label{step:ONCF:create-gadgets} For each $i,i' \in [k]$ and $j,j' \in [n]$ add a new ONCF-gadget $G^{i,i'}_{j,j'}$. Identify vertex $g_1$ of the gadget with $p_{i,j}$ and identify vertex $g_2$ of the gadget with $p_{i',j'}$. Add the edge $\{g_3, s^{i,i'}_{j,j'}\}$. Finally, connect vertex $g_{10}$ to $R$.
\end{enumerate}

In the remainder, we observe that $c(R) \neq c(B)$ for any $2$-ONCF-coloring of $G$. Thereby, we can safely rename the colors such that $c(R) = \red$ and $c(B) = \blue$.

\begin{myclaim}\label{claim:ONCF:fixed-colors}
Let $c$ be any $2$-ONCF-coloring of $G$, then $c(a_1) = c(a_2) = \red$, $c(v_j^1)=c(v_j^2)=\blue$ for all $j \in [n]$, and $c(h_i^1) = c(h^2_i) = \blue$ for all $i \in [k]$.
\end{myclaim}
\begin{claimproof}
This follows immediately from the fact that there is a degree-$2$ vertex connecting these vertices to $R$ or $B$ respectively.\qed
\end{claimproof}

\begin{myclaim}\label{claim:ONCF:select-y}
For any $2$-ONCF-coloring $c$ of $G$, there exists exactly one vertex $y^* \in Y$ such that $c(y^*) = \blue$ and for all other vertices $y \in Y \setminus \{y^*\}$, $c(y) =\red$.
\end{myclaim}
\begin{claimproof}
Observe that $N(a) = \{a_1,a_2\} \cup Y$ and that $c(a_1) = c(a_2) = \red$ by Claim \ref{claim:ONCF:fixed-colors}. Thereby, $a$ must have a unique \blue neighbor and this neighbor is in $Y$.\qed
\end{claimproof}

\begin{myclaim}\label{claim:ONCF:gadgets}
Let $G^{i,i'}_{j,j'}$ be a gadget and let $c$ be a $2$-ONCF-coloring of $G$.  Then $c(g_1) = c(g_2) = \red$ in this gadget implies that $c(g_3) = \blue$.
\end{myclaim}
\begin{claimproof}
Suppose $c(g_1) = c(g_2) = c(g_3) = \red$. It follows from Lemma \ref{lem:ONCF-gadget:all-red} that thereby $c(g_9) := \red$. Observe however that $N(g_{10}) = \{R,g_9\}$ by definition. Since both these vertices are \red, this contradicts the assumption that $c$ is a proper ONCF-coloring.\qed
\end{claimproof}

\begin{myclaim}\label{claim:ONCF:select-clique}
For any $2$-ONCF-coloring $c$ of $G$, there exist distinct $j_1,\ldots,j_k$ such that $c(p_{i,j_i}) = \red$ and for all other $p \in P$, $c(p) := \blue$.
\end{myclaim}
\begin{claimproof}
We start by showing that for each $i \in [k]$, there is exactly one $j \in [n]$ such that $c(p_{i,j}) = \red$. Consider the neighborhood of vertex $h_i$. $N(h_i):= \{h^1_i,h^2_i\}\cup \{p_{i,j} \mid j \in [n]\}$. Since $c(h^1_i) = c(h^2_i) = \blue$ by Claim~\ref{claim:ONCF:fixed-colors}, it follows that indeed $\{h^1_i,h^2_i\}\cup \{p_{i,j} \mid j \in [n]\}$ contains exactly one \red vertex, let this be vertex $p_{i,j_i}$. It remains to show that all $j_i$ are distinct.

We show this by proving that for each $j \in [n]$, there is at most one $i \in [k]$ such that $c(p_{i,j}) = \red$. Consider vertex $v_j$, observe that $N(v_j) := \{v^1_j,v^2_j,v^3_j\} \cup \{p_{i,j} \mid i \in [k]\}$. Since $c(v^1_j) = c(v^2_j) = \blue$ by Claim \ref{claim:ONCF:fixed-colors}, it follows that $v_j$ has a unique \red neighbor, and thus there is at most one $i \in [k]$ such that $c(p_{i,j}) = \red$. Hereby, the claim follows.\qed
\end{claimproof}

\begin{myclaim}\label{claim:oncf:if}
If there exists an instance $X_\ell$ that has a clique of size $k$, then $G$ can be $2$-ONCF-colored.
\end{myclaim}
\begin{claimproof}
Let $\ell$ be such that $X_\ell$ is a yes-instance for clique. Choose $j_1,\ldots,j_k \in [n]$ such that these vertices form a clique in $X_\ell$. We now give an ONCF-coloring $c$ for $G$, see Figure \ref{fig:cross-composition-ONCF} for an example ONCF-coloring of $G$.
\begin{enumerate}
\item Let $c(R) := c(R'):= \red$ and $c(B) := c(B') := \blue$.
\item Let $c(y_\ell) := \blue$. For all other vertices in $y \in Y$, let $c(y) := \red$.
\item Let $c(a) := \blue$ and $c(a_1) := c(a_2) := c(a'_1) := c(a'_2) := \red$.
\item Let $c(p_{i,j_i}) := \red$ for all $i \in [k]$. For all other vertices $p \in P$, let $c(p) := \blue$.
\item Let $c(h_i) := \red$ for all $i \in [k]$. Let $c(h) := \blue$ for all other vertices $h \in H$.
\item For $j \in [n]$, let $c(v^3_j):= \red$ if $j \notin \{j_1,\ldots,j_k\}$, let $c(v^3_j) := \blue$ otherwise. Let $c(v''^x_j):= \red$ for all $x \in [2]$. Let $c(v) := \blue$ for all remaining vertices $v \in V$.
\item Let $c(s) := \red$ for all $s \in S$.
\item It remains to color the introduced gadgets. Observe that vertices $g_1$ and $g_2$ of each gadget have already been colored, as they were identified with vertices from $P$. We now proceed as follows. Define $c(g_3) := \blue$ whenever $s^{i,i'}_{j,j'}$ has no \blue neighbor in $Y$ and define $c(g_3) := \red$ otherwise. Observe that since $j_1,\ldots,j_k$ form a clique in instance $X_\ell$, it never happens that $c(g_1) = c(g_2) = c(g_3) = \red$ by this definition. Color the remainder of each gadget using Lemma \ref{lem:ONCF-gadget:not-all-red}, such that the coloring satisfies property \ref{prop:fixed-colors} of the lemma statement.
\end{enumerate}
This defines a $2$-coloring of $G$, it remains to verify that $c$ is indeed a $2$-ONCF-coloring. We consider the neighborhood of each vertex in $G$.
\begin{enumerate}
\item $N(R) := \{R'\} \cup \{v'^1_j,v'^2_j \mid j \in [n]\} \cup \{h'^1_i,h'^2_i \mid i \in [n]\}$. Since $c(R') = \red$ and $c(x) = \blue$ for all $x \in \{v'^1_j,v'^2_j \mid j \in [n]\} \cup \{h'^1_i,h'^2_i \mid i \in [n]\}$, $N(R)$ is ONCF-colored. $N(B) := \{R',B',a_1',a_2'\}$, of which only $R'$ is red. Thus, $N(B)$ is ONCF-colored. Furthermore, $|N(B')| =1 $ and thereby it is trivially ONCF-colored, and $N(R') = \{R,B\}$ which have distinct colors as desired.
\item For any vertex $y \in Y$, $N(Y)$ contains vertex $a$ which is colored \blue. Furthermore, $N(Y) \setminus \{a\} \subseteq S$ and all vertices in $S$ are \red.
\item $N(a) = Y \cup \{a_1,a_2\}$. $Y$ contains exactly one \blue vertex, and $a_1$ and $a_2$ are colored \red. $N(a_1) := \{a,a_1'\}$, which have distinct colors as desired. Similarly, $N(a_2) := \{a,a_2'\}$ and these vertices are \blue and \red respectively. Finally, $N(a_1') := \{a_1,B\}$ and $N(a_2') := \{a_2,B\}$, it is easy to verify that these are ONCF-colored.
\item For $i \in [k]$ and $j \in [n]$, $N(p_{i,j})$ contains vertex $h_j$ which is \red. Furthermore, $N(p_{i,j}) \setminus \{h_j\}$ only contains vertices from $V$, which are \blue, and vertices $g_4$ and $g_5$ from numerous ONCF-gadgets, which are also \blue. Thereby it has \red as a unique color in its neighborhood.
\item For $i \in [k]$, $N(h_i) := \{p_{i,j} \mid j \in [n]\} \cup \{h^1_i,h^2_i\}$. Observe that all vertices in $N(h_i)$ are \blue, except vertex $p_{i,j_i}$ which is \red. For $x \in [2]$,  $N(h^x_i) := \{h'^x_i, h_i\}$ and these vertices receive distinct colors. $N(h'^x_i):= \{R, h^x_i\}$ and these vertices also receive distinct colors.
\item For $j \in [n]$, we observe that $v_j$ has exactly one \red neighbor in $P$ and all its other neighbors are \blue. Vertices $v_j^1$ and $v^j_2$ both have exactly one \red neighbor, namely vertex $v''^1_j$ or $v''^2_j$, respectively. Vertices $v'^x_j$ have one \blue and one \red neighbor for $x \in [2]$. The vertices $v''^x_j$ for $x\in[2]$ and vertex $v^3_j$ have degree one and are thus ONCF-colored by definition.
\item For $i,i'\in[k]$, $j,j' \in [n]$, vertex $s^{i,i'}_{j,j'}$ has neighbors in $Y$ and vertex $g_3$ in gadget $G^{i,i'}_{j,j'}$ . It follows from the definition of the coloring of $g_3$ and the fact that $Y$ has at most one \blue vertex that $s^{i,i'}_{j,j'}$ has exactly one \blue neighbor.
\item It remains to verify that all gadget vertices are ONCF-colored properly, consider the vertices of gadget $G^{i,i'}_{j,j'}$. The neighborhoods of vertices $g_4,g_5,\ldots,g_9$ are ONCF-colored by definition. Vertices $g_1$ and $g_2$ were identified with vertices from $P$ and have already been discussed above. $N(g_3) = \{g_5,s^{i,i'}_{j,j'}\}$ and these are \blue and \red, respectively. $N(g_{10}) = \{g_9, R\}$ and these are also \blue and \red.\qedhere
\end{enumerate}
\end{claimproof}

\begin{myclaim}\label{claim:oncf:only-if}
If $G$ can be $2$-ONCF-colored, then there exists $\ell \in [t]$ such that instance $X_\ell$ has a clique of size $k$.
\end{myclaim}
\begin{claimproof}
Let a $2$-ONCF-coloring $c$ of $G$ be given. By Claim \ref{claim:ONCF:select-y}, there exists $y \in Y$ with $c(y) := \blue$. Pick $\ell$ such that $y \in \{y_\ell,y'_\ell,y''_\ell\}$.  By Claim \ref{claim:ONCF:select-clique}, take distinct $j_1,\ldots,j_k \in [n]$ such that $c(p(i,j_i))=\red$. We will show that instance $X_\ell$ has a clique of size $k$, by proving that vertices $j_1,\ldots,j_k$ form a clique in $X_\ell$.

Suppose not, then there exist $i, i' \in [k]$ such that $j_i$ and $j_{i'}$ are not connected by an edge in $X_\ell$. We show that this leads to a contradiction. Consider gadget $G^{i,i'}_{j_i,j_{i'}}$. Vertices $g_1$ and $g_2$ of this gadget are colored \red, as they were identified with vertices $p_{i,j_i}$ and $p_{i',j_{i'}}$ respectively. It follows from Claim \ref{claim:ONCF:gadgets} that thereby $c(g_3) := \blue$ in this gadget. Now consider the neighborhood of vertex $s^{i,i'}_{j_i,j_{i'}}$. It contains vertex $g_3$ from gadget $G^{i,i'}_{j_i,j_{i'}}$ and the vertices $y_\ell,y'_\ell,y''_\ell$ since edge $\{j_{i},j_{i'}\}$ does not occur in instance $X_\ell$. It now follows that vertex $s^{i,i'}_{j_i,j_{i'}}$ has at least two \blue and two \red neighbors in $G$, which contradicts that $c$ is an ONCF-coloring of $G$.\qed
\end{claimproof}

From Claims \ref{claim:oncf:if} and \ref{claim:oncf:only-if}, it follows that $G$ can be $2$-ONCF-colored if and only if there exists an $\ell$ such that $X_\ell$ has a clique of size $k$. To prove the lower bound, it remains to bound the size of a vertex cover in $G$. Since $Y$ is an independent set in $G$, it follows that $V(G) \setminus Y$ is a vertex cover of $G$. Observe that
\[|V(G)\setminus Y| = \Oh(n^2 \cdot k^2).\]

To conclude, since we have given a cross-composition from \textsc{Clique} to $2$-\ONCF parameterized by vertex cover size, the lower bound now follows from Theorem~\ref{thm:cross-composition-implies-LB}.\qed
\end{proof}

Now, we use the lower bound obtained for $2$-\ONCF in Lemma~\ref{lem:ONCF-kernel-LB} and exhibit a polynomial parameter transformation to obtain the general lower bound for $q$-\ONCF for all $q \geq 2$. This completes the lower bound results for $q$-\ONCF claimed in Theorem~\ref{thm:noker}.

\begin{lemma}\label{lem:q-ONCF-kernel-LB}
For any $q \geq 2$, $q$-\ONCF parameterized by vertex cover size does not have a polynomial kernel, unless \containment.
\end{lemma}
\begin{proof}
We prove the result by giving a polynomial parameter transformation from $2$-\ONCF parameterized by vertex cover size to $q$-\ONCF parameterized by vertex cover size for any constant $q > 2$.  By~Theorem \ref{thm:ppt-works} and Lemma~\ref{lem:ONCF-kernel-LB}, this implies that $q$-\ONCF parameterized by vertex cover size,  does not have a polynomial kernel unless \containment for $q \geq 2$. We will do this by adding additional structures to the graph, that ensure that the original graph is colored using only $2$ colors, and that for any vertex in the original graph, its ONCF-color is also one of these two colors. Suppose we are given a graph $G$ for $2$-\ONCF, we show how to obtain $G'$ for $q$-\ONCF.
\begin{enumerate}
\item Start by initiating $G'$ as $G$. Let $V:= V(G)$.
\item\label{ONCF:step:X} Add $2q$ vertices $x^0_1,x^0_2,\ldots,x^0_q$ and $x^1_1,x^1_2,\ldots,x^1_q$. Let $X$ be the set of all these vertices. Add a clique on $x^0_1,x^0_2,\ldots,x^0_q$. Connect $x^1_j$ to $x^0_i$ for all $i \neq j$ with $i,j \in [q]$. Finally, subdivide all the edges between vertices in $X$. (Thus, vertices $x^0_1,x^0_2,\ldots,x^0_q$ form a subdivided clique in $G'$). Let the set of  vertices used to subdivide these edges be $X'$.
\item Add $2(q-2)$ vertices $y^0_\ell, y^1_\ell$ for $\ell \in [q-2]$, let $Y$ be the set containing all these vertices. Connect $y^0_\ell$ and $y^1_\ell$ to all vertices in $\{ x^0_i, x^1_i \mid i \in [q] \wedge i \neq \ell\}$. Then, connect $y^0_\ell$ to $x^0_\ell$ and connect $y^1_\ell$ to $x^1_\ell$.
 Finally, connect $y^0_\ell$ and $y^1_\ell$ to every vertex $v \in V$.
\item \label{ONCF:step:Y-to-palette} For $b\in\{0,1\}$ and $\ell \in [q-2]$, use a subdivided edge to connect $y^b_\ell$ to $x^0_j$ for all $j \neq \ell$. Let $Y'$ be the set containing all vertices used for subdividing these edges. This ensures that $y_\ell^b$ always receives color $\ell$.
\end{enumerate}

\begin{myclaim}
If $\chi_{\sf ON}(G') \leq q$, then $\chi_{\sf ON}(G) \leq 2$.
\end{myclaim}
\begin{claimproof}
Suppose $G'$ has a $q$-ONCF-coloring $c' \colon V(G') \rightarrow [q]$, we will now show that $G$ has an $2$-ONCF-coloring.
It is easy to observe that $c'(x^0_i) = c'(x^1_i)$ for all $i$, and furthermore $c'(x^0_i) \neq c'(x^0_j)$ for all $i \neq j \in [q]$, by the subdivided edges introduced in Step \ref{ONCF:step:X}.
By this observation,  we may assume without loss of generality that $c'(x^0_i) = c'(x^1_i) = i$.
It is easy to observe using Step \ref{ONCF:step:Y-to-palette} of the construction, that in such a coloring $c'(y^0_\ell) = c'(y^1_\ell) = \ell$. Thereby, for every $\ell \in [q-2]$, $N(v)$ contains two vertices of color $\ell$, for all $v \in V$. This implies that for any $\ell \in [q-2]$ and $v \in V$, we know that $\ell$ is not the color that ensures that $N(v)$ is $q$-ONCF-colored.

 Furthermore, $N(y^0_\ell)$ contains two vertices of color $i$ for all $i \neq \ell$ (namely $x^0_i$ and $x^1_i$), and one vertex of color $\ell$. Thereby, no vertex in $v$ can have color $\ell$.
This implies that only two colors are used in $V$, namely $q$ and $q-1$. We conclude that the coloring $c$ restricted to vertices in $V$ is a $2$-ONCF-coloring for $G$, after renaming the colors to $\{1,2\}$.\qed
\end{claimproof}

\begin{myclaim}
If $\chi_{\sf ON}(G) \leq 2$, then $\chi_{\sf ON}(G') \leq q$.
\end{myclaim}
\begin{claimproof}
Suppose $G$ has a $2$-\ONCF-coloring $c$, we show how to $q$-ONCF-color $G'$. First of all, let $c'(x^0_i) = c'(x^1_i) = i$ for all $i \in [q]$ and let $c'(y^0_\ell) = c'(y^1_\ell) = \ell$ for all $\ell \in [q-2]$. For $v \in V$, let $c'(v) = q-1$ when $c(v) = 1$ and let $c'(v) = q$ otherwise. For the vertex $x^0$ on the subdivided edge from $x^0_q$ to $x^1_{q-1}$, let $c(x^0) = q-1$. Similarly, for the vertex $x^1$ on the subdivided edge from $x^1_q$ to $x^0_{q-1}$, let $c(x^1) = q-1$.
For all remaining vertices, let $c'(v) := q$. It remains to show that this gives a $q$-ONCF-coloring of $G'$. See Figure \ref{ONCF-poly-param-transformation} for a sketch of $G'$.
\begin{figure}
\centering
\includegraphics{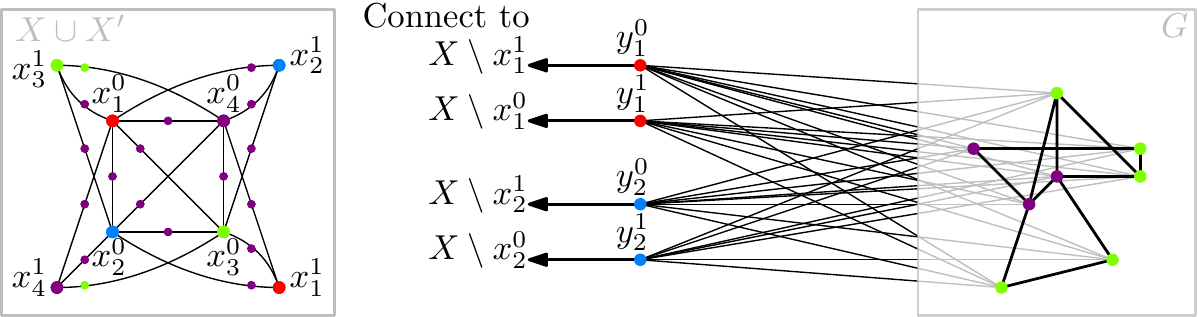}
\caption{Reduction from $G$ to $G'$ for $q = 5$. The subdivided edges from $Y$ to $X$ added in Step \ref{ONCF:step:Y-to-palette} of the construction are omitted for simplicity.}
\label{ONCF-poly-param-transformation}
\end{figure}
\begin{enumerate}
\item The neighborhoods of vertices in $V$ are trivially $q$-ONCF-colored since they have no vertices of color $q$ or $q-1$ outside $V$, and $c$ was a valid $2$-ONCF-coloring.
\item For $i \in [q-2]$, $b \in \{0,1\}$ vertex $x^b_i$ has exactly one neighbor of color $i$, namely $y^b_i$. For $i \in \{q-1,q\}$, vertices in $x^b_i$ have exactly one neighbor of color $q-1$, namely $x^b$ or $x^{(1-b)}$. Vertices in $X'$ have degree $2$ and have neighborhoods that are $q$-ONCF-colored by this definition.

\item Vertices $y_\ell^0$ and $y_\ell^1$ for $\ell \in [q-2]$ each have exactly one neighbor of color $\ell$ and are thus $q$-ONCF-colored.
\item Neighborhoods of the vertices used for subdividing edges always have two vertices of distinct color, by definition.\qedhere
\end{enumerate}
\end{claimproof}

Since $G'$ is a copy of $G$ to which we add $\Oh(q)$ additional vertices, it follows that the vertex cover of $G'$ is bounded by $k + \Oh(q) = \Oh(k)$ where $k$ is the size of a vertex cover in $G$. Thus, we have given a polynomial-parameter transformation from $2$-\ONCF to $q$-\ONCF with both problems having vertex cover size as parameter, and the theorem statement follows from Theorem~\ref{thm:ppt-works} and Lemma~\ref{lem:ONCF-kernel-LB}.\qed
\end{proof}

\subsection{Kernel lower bound for CNCF-Coloring}\label{subsec:cncf-col}
In this part, we complete the proof of Theorem~\ref{thm:noker} by showing that $q$-\CNCF parameterized by vertex cover size has no polynomial kernel when $q$ is at least $3$.  To do this, we first introduce a useful gadget, which will serve as a color palette in our lower bound construction. The gadget is based on the graphs $G_k$ defined by Abel et al. \cite[Section 3.1]{AbelADFGHKS2017}.  

\begin{definition}[{\cite{AbelADFGHKS2017}}]\label{def:G_k}
 For every positive integer $k$, a graph $G_k$ is recursively defined  as follows:
\begin{enumerate}
\item $G_1$ consists of a single isolated vertex. $G_2$ is a $K_{1,3}$ with one edge subdivided by another vertex (refer also to Figure \ref{fig:cncf-palette}).
\item Given $G_k$ and $G_{k-1}$, $G_{k+1}$ is constructed as follows for $k \geq 2$:
    \begin{itemize}
    \item Take a complete graph $G = K_{k+1}$ on $k+1$ vertices.
    \item To each vertex $v \in V(K_{k+1})$, attach two disjoint and independent copies of $G_k$, adding an edge from $v$ to every vertex of both copies of $G_k$.
    \item For each edge $e = \{v,w\} \in E(K_{k+1})$, add two disjoint and independent copies of $G_{k-1}$, adding an edge from $v$ and $w$ to every vertex of both copies.
    \end{itemize}
\end{enumerate}
\end{definition}

Let $\chi^*_{\text{\sf CN}}(G)$ denote the minimum number of colors needed to CNCF-color $G$, when it is allowed to not color certain vertices.
Abel et al. have shown the following lemma.
\begin{lemma}[{\cite[Lemma 3.3]{AbelADFGHKS2017}}]\label{lem:palette:gK}
For $G_k$ constructed in this manner, $\chi^*_{{\sf CN}}(G_k) = k$.
\end{lemma}

We use this to define the palette-gadget $C_q$.

\begin{definition}\label{def:cncf-palette}
To create a \emph{palette-gadget} $C_q$ start from a complete graph on vertices $c_1,\ldots,c_q$. Then, add vertices $c_i'$ for $i \in [q]$ and connect $c_i'$ to $c_j$ for all $i \neq j$. Let $D:= \{c_i,c_i'\mid i \in [q]$ be the set of distinguished vertices of the gadget. Finally, for each edge $\{u,v\} \in E(C_q)$ with $u,v \in D$, add two new distinct copies of $G_{q-1}$ and connect all vertices in these copies to both $u$ and $v$. See Figure \ref{fig:cncf-palette} for an example of the palette-gadget $C_3$.
\end{definition}

\begin{figure}[t]
\centering
\includegraphics{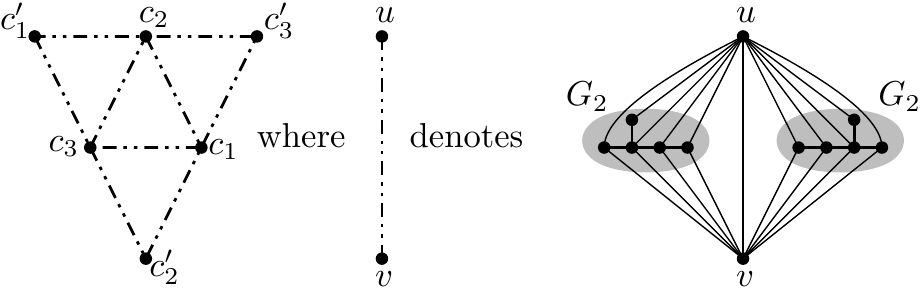}
\caption{A palette-gadget (left) where every dashed edge should be interpreted as the gadget depicted on the right.}
\label{fig:cncf-palette}
\end{figure}
The next two lemmas are used to establish that a palette gadget can indeed serve as a color palette for CNCF-Coloring.

\begin{lemma}\label{lem:CNCF:palette-is-palette}
Let $G$ be a graph and let $C$ be a set of vertices such that $G[C]$ is isomorphic to the palette-gadget $C_q$ for some $q \geq 3$. Let $f$ be a $q$-CNCF-coloring of $G$. Then $f(c_i) \neq f(c_j)$ for all $i,j\in[q]$ with $i\neq j$. Furthermore,  $f(c_i) = f(c_i')$ for all $i \in [q]$.
\end{lemma}
\begin{proof}

We show this by showing that if $\{u,v\}$ is an edge such that there are two distinct copies of $G_{q-1}$, say $G_{q-1}^1$ and $G_{k-1}^2$, that are connected to both $u$ and $v$ and no other vertices in the graph, then $f(u) \neq f(v)$ in any $q$-CNCF-coloring of $G$. The results of the lemma statement then follow from the definition of the palette.

Suppose that for contradiction that there exists a $q$-CNCF-coloring $f$ with $f(u) = f(v)$. Let $i$ be a color such that $|\{f(u) = i \mid u \in N[v]\}| =1$. Observe that $i \neq f(u)$ as $f(u) = f(v)$ and $u,v \in N[v]$. Therefore, either $G_{q-1}^1$ or $G^2_{q-1}$ does not use color $i$, w.l.o.g. let this be $G_{q-1}^1$. We show that thereby $\chi^*_{\text{\sf CN}}(G^1_{q-1}) = q-2$, which contradicts Lemma \ref{lem:palette:gK}.

Define partial coloring $f^*$ of $G^1_{q-1}$ as follows. For any $x \in V(G^1_{q-1})$ with $f(x) \neq f(u)$, let $f^*(x) := f(x)$. For any $x \in V(G^1_{q-1})$ with $f(x) = f(u)$, leave $f^*(x)$ undefined. Observe that hereby, the range of $f^*$ is a subset of $[q] \setminus \{i,f(u)\}$ and thus $f^*$ defines a $(q-2)$-coloring of $G^1_{q-1}$. From the correctness of $f$ and the fact that any vertex in $G^1_{q-1}$ has two neighbors of color $f(v) = f(u)$ under $f$, it follows that $f^*$ is a partial $(q-2)$-CNCF-coloring of $G^1_{q-1}$, which is a contradiction.
\qed\end{proof}

\begin{lemma}\label{lem:CNCF:palette-colorable}
Let $C_q$ be a palette-gadget for $q \geq 3$. Then there exists a $q$-CNCF-coloring $f \colon V(C_q) \rightarrow [q]$  such that
\begin{enumerate}
\item $f(c_i) = f(c_i') = i$ for all $i \in [q]$, and
\item for all $i\in[q]$, $N[c_i]$ contains exactly one vertex of color $i$.
\end{enumerate}
\end{lemma}
\begin{proof}
We start by defining $f(c_i) := f(c_i') := i$ for all $i\in[q]$. Let $D:= \{c_i,c_i'\mid i \in [q]\}$, observe that the color of all vertices in $D$ has now been defined. All vertices in $V(C_q) \setminus D$ induce distinct copies of $G_{q-1}$. Consider an arbitrary copy of $G_{q-1}$ in $C_q$, and suppose it was added to the palette gadget for the edge $\{d,d'\}$ with $d,d' \in D$. Color all vertices of this $G_{q-1}$ with a color $x \in [q]$ such that $f(d) \neq x \neq f(d')$. Observe that such a color exists since $q \geq 3$.

Both requirements are satisfied by the definition of $f$, it remains to show that $f$ is a $q$-CNCF-coloring of $C_q$. Vertices $c_i$ and $c_i'$ have color $i$ and have no neighbors of color $i$, and are thereby properly CNCF-colored. Vertices in $V(G)\setminus D$ are distinct copies of $G_{q-1}$. It is easy to verify that in $N(G_{q-1})$ there are two vertices with a unique color, corresponding to the edge for which it was added. Since these colors are not used to color $G_{q-1}$, the result follows.
\qed\end{proof}

Using the gadget introduced above, we now prove the kernelization lower bound.

\begin{lemma}\label{lem:3-CNCF-kernel-LB}
For any $q\geq 3$, $q$-\CNCF parameterized by vertex cover size does not have a polynomial kernel, unless \containment.
\end{lemma}
\begin{proof}
To prove this theorem, we will give a cross-composition starting from \textsc{Clique} that is very similar to the one given in Lemma \ref{lem:ONCF-kernel-LB}.

\begin{figure}[t]
\centering
\includegraphics{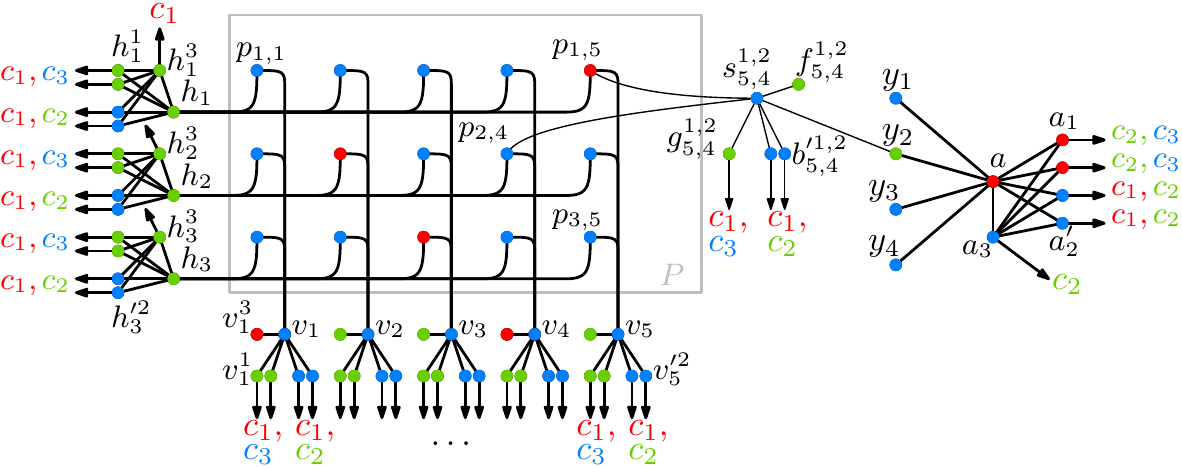}
\caption{A sketch of the constructed graph $G$ for $q=3$, $k = 3$, $n = 5$, and $t = 4$. All vertices created in Steps \ref{CNCF:step:S} and \ref{CNCF:step:Sgadgets} are omitted, except those created with $i = 1$, $i' = 2$, $j = 5$ and $j' = 4$, assuming edge $\{4,5\}$ is not present in instance $X_2$. The palette $C$ is omitted, edges to the palette are drawn as arrows.}
\label{fig:cross-composition-CNCF}
\end{figure}
We define the same polynomial equivalence relation; let two instances of \textsc{Clique} be equivalent if the graphs have the same number of vertices and they ask for a clique of the same size. Suppose we are given $t$ instances of clique that are equivalent under this relation, labeled $X_1,\ldots,X_t$. Let every instance have $n$ vertices and ask for a clique of size $k$. We enumerate the vertices in each instance arbitrarily as $1,\ldots,n$. We create an instance $G$ for $q$-\CNCF by the following steps, refer to Figure \ref{fig:cross-composition-CNCF} for a sketch of $G$.
\begin{enumerate}
\item \label{CNCF:step:palette} Create a palette-gadget $C_q$ with distinguished vertices $c_1,\ldots,c_q$ and $c_1',\ldots,c_q'$ by Definition \ref{def:cncf-palette}.
\item \label{CNCF:step:Y} Create $t$ vertices $y_1,\ldots,y_t$ and let $Y := \{y_\ell \mid \ell \in [t]\}$. The idea is that exactly one of these vertices $y_\ell$ will receive color $2$ in any $q$-CNCF-coloring, and this indicates that $X_\ell$ is a yes-instance for clique.
\item \label{CNCF:step:a} Add a vertex $a$ and connect $a$ to $y_\ell$ for all $\ell \in [t]$. Create vertices $a_1$ and $a_1'$ and connect both of these to vertices $c_2$ and $c_3$ in the palette-gadget. Furthermore, create vertices $a_2$ and $a_2'$ and connect them to vertices $c_1$ and $c_2$.
    Finally, create vertex $a_3$, connect $a_3$ to $c_2$, $a_1$, $a_1'$, $a_2$, and $a_2'$. Connect $a$ to $a_1$, $a_1'$, $a_2$, $a_2'$, and $a_3$.
    The idea is that in any $q$-CNCF-coloring, vertices $a_1$ and $a_1'$ receive the color of $c_1$, $a_2$ and $a_2'$ receive the color of $c_3$ and that $a_3$ has exactly one neighbor of color $2$ and two neighbors of both remaining colors, implying that the color of $a$ cannot be $2$.
\item \label{CNCF:step:P} Add vertices $p_{i,j}$ for all $i \in [k]$, $j \in [n]$. Let $P$ be the set containing all of these vertices. The idea is that vertices in $P$ receive colors $1$ and $3$, such that for every $i$ there is exactly one vertex $p_{i,j}$ of color $1$. The vertices of color $1$ will correspond to the vertices that form a clique in one of the input instances.
\item \label{CNCF:step:H} Add a vertex $h_i$ for all $i \in [k]$ and connect $h_i$ to $p_{i,j}$ for all $j \in [n]$. Add vertices $h^1_i$, $h'^1_i$, $h^2_i$, $h'^2_i$ and $h^3_i$.  Add edges $\{c_1,h^1_i\}$, $\{c_1,h'^1_i\}$, $\{c_3,h^1_i\}$, $\{c_3,h'^1_i\}$,  $\{c_1,h^2_i\}$, $\{c_1,h'^2_i\}$, $\{c_2,h^2_i\}$,  and $\{c_2,h'^2_i\}$. For all $i \in [k]$, add a vertex $h^3_i$ and connect it to $c_1$. Connect $h^3_i$ to $h^1_i$, $h'^1_i$, $h^2_i$, and $h'^2_i$. Finally, connect $h_i$ to $h^1_i$, $h'^1_i$, $h^2_i$, $h'^2_i$, and $h^3_i$. Let $H$ be the set of all vertices created in this step. These vertices will ensure that for each $i$, there is exactly one vertex $p_{i,j}$ of color $1$.
\item \label{CNCF:step:V} Add a vertex $v_j$ for all $j \in [n]$ and connect $v_j$ to $p_{i,j}$ for all $i \in [k]$. Add vertices $v^1_j$ and $v'^1_j$ and connect them to $c_1$ and $c_3$. Add vertices $v^2_j$ and $v'^2_j$ and connect these to $c_1$ and $c_2$. Add a vertex $v^3_j$. Finally, connect vertex $v_j$ to $v^1_j$, $v'^1_j$, $v^2_j$, $v'^2_j$, and $v^3_j$. The vertices added in this step ensure that there cannot be $i,i'\in [k]$ such that both $p_{i,j}$ and $p_{i',j}$ receive color $1$ for some $j \in [n]$.
\item \label{CNCF:step:S} For each $i,i' \in [k]$ and $j,j' \in [n]$, add vertex $s^{i,i'}_{j,j'}$, let the set containing all these vertices be $S$. Connect $s^{i,i'}_{j,j'}$ to $p_{i,j}$ and $p_{i',j'}$. Furthermore, connect $s^{i,i'}_{j,j'}$ to $y_\ell$ whenever $\{j,j'\}$ is \emph{not} an edge in instance $X_\ell$. These vertices are used to verify whether the vertices selected by $P$ indeed form a clique in the selected input instance.
\item \label{CNCF:step:Sgadgets} For each $i,i' \in [k]$ and $j,j' \in [n]$, add vertices $f^{i,i'}_{j,j'}$, $g^{i,i'}_{j,j'}$, $b^{i,i'}_{j,j'}$, and $b'^{i,i'}_{j,j'}$, and connect all these vertices to $s^{i,i'}_{j,j'}$. Connect $g^{i,i'}_{j,j'}$ to $c_1$ and $c_3$, connect $b^{i,i'}_{j,j'}$ to $c_1$ and $c_2$, and finally connect $b'^{i,i'}_{j,j'}$ to $c_1$ and $c_2$.
\item \label{CNCF:step:connect-to-palette} For every vertex in $v \in V(G) \setminus V(C_q)$, add the edges $\{v,c_i\}$ and $\{v,c_i'\}$ for all $3 < i \leq q$. Thus, we connect every non-palette vertex in $G$ to all but the first  three colors from the palette. This step ensures that colors $i > 3$ are not used to color $V(G)\setminus C_q$.
\end{enumerate}

It follows from Lemma \ref{lem:CNCF:palette-is-palette} that all $c_i$ receive distinct colors. Therefore, we will from now on assume that $c(c_i) = i$ for any coloring $c$. 
Furthermore, we observe that for $i \in [q]$, vertex $c_i$ is connected to $c'_j$ and $c_j$ for all $j \in [q]\setminus \{i\}$. It follows that vertex $c_i$ has its own color (if any) as its CNCF-color, since it is connected to two vertices of all remaining colors.

In the proofs of the remaining claims, we will regularly use that any non-palette vertex in $G$ has two neighbors of color $i$ for all $i > 3$.

\begin{myclaim}\label{claim:CNCF:select-y}
For any $q$-CNCF-coloring $c$ of $G$, there exists exactly one vertex $y^* \in Y$ such that $c(y^*) = 2$.
\end{myclaim}
\begin{claimproof}
It follows from the observation above, that $c(a_1) = c(a'_1) = 1$ and $c(a_2) = c(a'_2) = 2$. Furthermore, $c(a_3) \neq 2$.  Thereby, $N[a_3]$ contains vertex $a$, together with one vertex of color $2$ and two vertices of color $i$ for all $i \neq 2$, implying $c(a) \neq 2$. It follows that $N[a]$ contains at least two vertices of color $1$ and two of color $3$ and that $N[a] \setminus Y $ contains no vertices of color $2$. Thereby, $N[a] \cap Y = Y$ must have exactly one vertex of color $2$.
\qed\end{claimproof}

\begin{myclaim}\label{claim:CNCF:select-clique}
For any $3$-CNCF-coloring $c$ of $G$, there exist distinct $j_1,\ldots,j_k$ such that $c(p_{i,j_i}) = 1$ and for all other $p \in P$, $c(p) \neq 1$.
\end{myclaim}
\begin{claimproof}
We start by showing that for each $i \in [k]$, there exists $j_i \in [n]$ such that $c(p_{i,j_i}) = 1$. We will then show that these $j_i$ are indeed distinct.

Let $i\in[k]$, then $\{h^1_i,h'^1_i,h^2_i,h'^2_i\}\subseteq N[h_i]$ and thus $N[h_i]$ contains two vertices of colors $2$ and $3$. Since $h_i$ is connected to $h_i^3$ and $N[h_i^3]$ contains at least one vertex of color $1$, and two vertices of both color $2$ and $3$,  it follows that $c(h_i) \neq 1$. Thereby, the CNCF-color for $h_i$ is $1$ and thus there exists a unique vertex in $\{p_{i,j}\mid j \in [n]\}$ that receives color $1$.

It remains to show that these $j_i$ are indeed distinct. We do this by showing that there cannot be vertices $p_{i,j}$ and $p_{i',j}$ such that $c(p_{i,j}) = c(p_{i',j}) = 1$. Suppose for contradiction that there are $j \in [n]$ and $i,i' \in [k]$ such that $c(p_{i,j}) = c(p_{i',j}) = 1$. But then $N[v_j]$ contains  vertices $v_j^1$ and $v'^1_j$ that have color $2$,  vertices $v^2_j$
 and $v'^2_j$ of color $3$, and the aforementioned two  vertices of color $1$. Since it furthermore contains two vertices of color $i$ for all $i \geq 4$, this contradicts that $c$ is a CNCF-coloring for $G$.\qed
 \end{claimproof}

\begin{myclaim}\label{claim:CNCF:if}
If there exists $\ell \in [t]$ such that $X_\ell$ has a clique of size $k$, then $G$ is $q$-CNCF-colorable.
\end{myclaim}
\begin{claimproof}
Take $\ell$ such that $X_\ell$ is a yes-instance for \textsc{Clique} and let $j_1,\ldots,j_k$ be such that vertices $\{j_1,\ldots,j_k\}$ form this clique in instance $X_\ell$. We give a coloring $c:V(G) \rightarrow [q]$ of $G$. We start by showing how to color the vertices defined in each step of the construction, this coloring is also depicted in Figure \ref{fig:cross-composition-CNCF}.
\begin{enumerate}
\item We start by coloring the palette $C$ as in Lemma \ref{lem:CNCF:palette-colorable}, such that $c(c_i) := c(c_i') := i$.
\item Let $c(y_\ell) := 2$ and $c(y) := 3$ for all other vertices $y \in Y$.
\item Let $c(a) := c(a_1):=c(a'_1) := 1$ and $c(a_2) := c(a'_2) := c(a_3) := 3$.
\item For all $i\in [k]$, let $c(p_{i,j_i}) := 1$. For all other $p \in P$, let $c(p) := 3$.
\item For all $i \in [k]$, let $c(h_i) := c(h^1_i) := c(h'^1_i) := c(h^3_i) := 2$ and let $c(h^2_i) := c(h'^2_i) := 3$.
\item For all $j\in[n]$, let $c(v_j) := 3$. For all $j\in[n]$, let $c(v^1_j):=c(v'^1_j) := 2$ and let $c(v^2_j):=c(v'^2_j) := 3$.  Let $c(v^3_j) := 2$ if there exists $i \in [k]$ such that $j_i = j$. Let $c(v^3_j) := 1$ otherwise.
\item For $i,i' \in [k]$ and $j,j' \in [n]$, let $c(s^{i,i'}_{j,j'}):=3$.
\item For $i,i' \in [k]$ and $j,j' \in [n]$, let $c(g^{i,i'}_{j,j'}):=2$, $c(b^{i,i'}_{j,j'}):=c(b'^{i,i'}_{j,j'}):=3$. Finally, if $s^{i,i'}_{j,j'}$ at this point has no neighbor of color $1$, let $c(f^{i,i'}_{j,j'}):=1$. Furthermore, if $s^{i,i'}_{j,j'}$ is not connected to $y_\ell$ (meaning $\{j,j'\}$ is an edge in $X_\ell$), define $c(f^{i,i'}_{j,j'}):=1$. Otherwise, let $c(f^{i,i'}_{j,j'}):=2$.
\end{enumerate}
It remains to show that this indeed gives a $q$-CNCF-coloring of $G$. We verify this for all vertices. $C$ is CNCF-colored by the fact that $c_i$ and $c_i'$ are colored by their own color and not connected to any other vertex of color $i$. Vertices in $Y$ are CNCF-colored by vertex $a$ which has color $1$. $N[a]$ contains exactly one vertex of color $2$, namely $y_{\ell}$. $N[a_1],N[a_1'],N[a_2],N[a_2']$, and $N[a_3]$ all contain exactly one vertex of color $2$, namely $c_2$. For $p \in P$, $N[p]$ contains exactly one vertex $h_i$ of color $2$ and no other vertices of color $2$. For all $i \in [k]$, $N[h_i]$ contains exactly one vertex $p_{i,j_i}$ of color $1$ and not other vertices of color $1$. Vertices $h_i^1,h'^1_i,h_i^2,h'^2_i$, and $h_i^3$ have $c_1$ as their unique neighbor with color $1$. Similarly, for all $j\in[n]$, the vertex $v_j$ has exactly one neighbor of color $1$ from the set $\{p_{i,j} \mid i \in [k]\}\cup \{v_i^3\}$. Vertices $v_j^1,v'^1_j,v^2_j$, and $v'^2_j$ have $c_1$ as their only neighbor of color $1$. Vertex $v_j^3$ has a distinct color from its only neighbor $v_j$ and thereby its closed neighborhood is CNCF-colored. Since $c(f_{j,j'}^{i,i'}) \in \{1,2\}$ and $c(s_{j,j'}^{i,i'}) = 3$ for all $i,i'\in[k]$, $j,j' \in [n]$, vertex $f_{j,j'}^{i,i'}$ receives a different color than its only neighbor. Vertices $g^{i,i'}_{j,j'}$, $b^{i,i'}_{j,j'}$, and $b'^{i,i'}_{j,j'}$ all have a unique neighbor of color $1$, namely vertex $c_1$. Finally we check the closed neighborhood of vertices  $s^{i,i'}_{j,j'}$ for $i,i'\in[k]$, $j,j'\in[n]$. If $s^{i,i'}_{j,j'}$ is not connected to $y_\ell$, it is ensured that it has exactly one neighbor of color $2$, namely vertex $g^{i,i'}_{j,j'}$. Otherwise, observe that $c(p_{i,j}) \neq 1$ or $c(p_{i',j'}) \neq 1$ as $\{j,j'\}$ is not an edge in $X_\ell$. The choice of coloring for $f^{i,i'}_{j,j'}$ ensures that in this case, $s^{i,i'}_{j,j'}$ has a unique neighbor of color $1$.
\qed\end{claimproof}

\begin{myclaim}\label{claim:CNCF:iff}
If $G$ has a $q$-CNCF-coloring, then there exists $\ell \in [t]$ such that $X_\ell$ has a clique of size $k$.
\end{myclaim}
\begin{claimproof}
Let $c$ be a CNCF-coloring of $G$. It follows from Claim \ref{claim:CNCF:select-y} that there exists a vertex $y \in Y$ with $c(y) = 2$. Let $\ell$ be such that $c(y_\ell) = 2$. We show that $X_\ell$ has a clique of size $k$. By Claim \ref{claim:CNCF:select-clique}, there exist distinct $j_1,\ldots,j_k$ such that $c(p_{i,j_i}) = 1$. We show that the vertices $j_1,\ldots,j_k$ form the desired clique in $X_\ell$.

Suppose for contradiction that there are distinct $i,i' \in [k]$ such that $\{j_i,j_{i'}\}$ is not an edge in instance $X_\ell$. We will show that $N[s^{i,i'}_{j_i,i_{i'}}]$ is not properly CNCF-colored. First of all, $N[s^{i,i'}_{j_i,i_{i'}}]$ contains the two vertices $p_{i,j_i}$ and
$p_{i',j_{i'}}$ with $c(p_{i,j_i})=c(p_{i',j_{i'}})=1$. Furthermore it contains the two vertices $b^{i,i'}_{j_i,i_{i'}}$ and $b'^{i,i'}_{j_i,i_{i'}}$ that have color $3$, and finally it contains two vertices of color $2$, namely $g^{i,i'}_{j_i,i_{i'}}$ and $y_\ell$. Furthermore, $N[s^{i,i'}_{j_i,i_{i'}}]$ contains two vertices of color $i$ for all $i > 3$, by Step \ref{CNCF:step:connect-to-palette} of the construction. This however contradicts that $c$ is a CNCF-coloring of $G$, and thus we conclude that $j_1,\ldots,j_k$ form a clique of size $k$ in instance $X_\ell$.
\qed\end{claimproof}

It follows from Claims \ref{claim:CNCF:if} and \ref{claim:CNCF:iff} that $G$ has a $q$-CNCF-coloring if and only if one of the given input instances was a yes-instance for \textsc{Clique}. It remains to bound the size of a vertex cover in $G$, to conclude the cross-composition. It is easy to verify that $V(G) \setminus Y$ is a vertex cover for $G$, since $Y$ is an independent set. Thereby the size of a vertex cover in $G$ is at most $|V(G)\setminus Y| = \Oh(n^2k^2)+f(q)$, where $f(q)$ is the size of palette-gadget $C_q$. As this is properly bounded for a cross-composition, the theorem statement follows from Theorem~\ref{thm:cross-composition-implies-LB}.
\qed\end{proof}

\subsection{Generalized kernel for $2$-CNCF-Coloring}
\label{secappen:poly-kernel-2CNCF}
In this part we prove Theorem~\ref{thm:genker}, by obtaining a polynomial generalized kernel for $2$-\CNCF parameterized by vertex cover size. This result is in contrast to the kernelization results we obtain for $q$-\CNCF for $q \geq 3$ as well as $q$-\ONCF for $q \geq 2$. We will start by transforming an instance of $2$-\CNCF to an equivalent instance of another problem, namely \drootCSP. We will then carefully rephrase the \drootCSP instance such that it uses only a limited number of variables, such that we can use a known kernelization result for \drootCSP to obtain our desired compression. We start by introducing the relevant definitions.

Define \drootCSP over a field $F$ as follows \cite{JansenP18}.

\defproblem{\textsc{\drootCSP}}
{A list $L$ of polynomial equalities over variables $V = \{x_1, \ldots, x_n\}$. An equality is of the form $f(x_1, \ldots, x_n) = 0$, where $f$ is a multivariate polynomial over $F$ of degree at most~$d$.}
{Does there exist an assignment of the variables $\tau \colon V \to \{0,1\}$ satisfying all equalities (over $F$) in $L$?
}

A field $F$ is said to be \emph{efficient} if both the field operations and Gaussian elimination can be done in polynomial time in the size of a reasonable input encoding. In particular, $\mathbb{Q}$ is an efficient field by this definition. The following theorem was shown by Jansen and Pieterse.

\begin{theorem}[{\cite[Theorem 5]{JansenP18}}]\label{thm:subset_kernel_Q}
There is a polynomial-time algorithm that, given an instance $(L,V)$ of \drootCSP over an efficient field $F$, outputs an equivalent instance~$(L',V)$ with at most $n^d + 1$ constraints such that $L' \subseteq L$.
\end{theorem}

Using the theorem introduced above, we can now prove Theorem~\ref{thm:genker}.

\begin{proof}[Proof of Theorem~\ref{thm:genker}]
Given an input instance $G$ with vertex cover $S$ of size $k$, we start by preprocessing $G$. For each set $X \subseteq S$ with $|X| \leq 2$, mark  $3$ vertices in $v \in G \setminus S$ with $N(v) = X$ (if there do not exist $3$ such vertices, simply mark all). Let $S' \subseteq V(G) \setminus S$ be the set of all marked vertices. Remove all $w \in V(G)\setminus (S \cup S')$ with $deg(w) \leq 2$ from $G$. Let the resulting graph be $G'$.
\begin{myclaim}\label{claim:kernel:CNCF:preprocess}
$G'$ is $2$-CNCF-colorable if and only if $G$ is $2$-CNCF-colorable.
\end{myclaim}

\begin{claimproof}
In one direction, suppose $G'$ has a $2$-CNCF coloring $c$ using colors $\{r,b\}$. Consider a vertex $w \in V(G) \setminus V(G')$. Let $X_w \subseteq S$ be the neighborhood of $w$. Note that $\vert X_w \vert$ is at most $2$. Consider $N(X_w) \cap S'$. Since $w$ was deleted, there are $3$ vertices in $N(X_w) \cap S'$. Consider the color from $\{r,b\}$ that appears in majority on the vertices of $N(X_w) \cap S'$. If we color $w$ with the same color, it is easy to verify that this  extension of $c$ to $G$ is a $2$-CNCF coloring of $G$.

In the reverse direction, suppose $G$ has a $2$-CNCF coloring $c$ using colors $\{r,b\}$. We describe a new coloring $c'$ for $G$ as follows. Consider a subset $X \subseteq S$ of size at most $2$ and let $N$ be the set of vertices in $G \setminus S$ that have $X$ as their neighborhood. If $\vert N \vert > 3$ and $N \setminus S'$ has a vertex $w$ that is uniquely colored in the set $N$, then we arbitrarily choose a vertex $w' \in N \cap S'$. We define $c'(w')=c(w)$ and $c'(w)=c(w')$. All other vertices have the same color in $c$ and $c'$. It is easy to verify that $c'$ is also a $2$-CNCF coloring of $G$ and the restriction of $c'$ to $G'$ is a $2$-CNCF coloring of $G'$.\claimqed
\end{claimproof}

We continue by creating an instance of  \rootCSP{2} that is satisfiable if and only if $G'$ is $2$-CNCF-colorable. Let $V := \{r_v,b_v \mid v \in V(G)\}$ be the variable set. We create $L$ over $\mathbb{Q}$ as follows.
\begin{enumerate}
\item \label{step:equations:red-or-blue} For each $v \in V(G')$, add the constraint $r_v + b_v - 1 = 0$ to $L$.
\item \label{step:equations:per-vertex} For all $v \in V(G')$, add the constraint
$(-1 +\sum_{u \in N[v]} r_v )\cdot (-1+\sum_{u \in N[v]} b_v)=0.$
\item \label{step:equations:ensure-nbhood} For each $v \in V(G')\setminus (S \cup S')$ of degree $d_v = |N(v)|$ add the constraint
    \[(\sum_{u \in N(v)} r_u)(-1 + \sum_{u \in N(v)} r_u)(-(d_v-1) + \sum_{u \in N(v)} r_u)(-d_v + \sum_{u \in N(v)} r_u) = 0.\]

    Note that such a constraint is a quadratic polynomial.
\end{enumerate}
Intuitively, the first constraint ensures that every vertex is either \red or \blue. The second constraint ensures that in the closed neighborhood of every vertex, exactly one vertex is \red or exactly one is \blue. The third constraint is seemingly redundant, saying that the open neighborhood of every vertex outside the vertex cover does not have two \red or two \blue vertices, which is clearly forbidden.  The requirement for these last constraints is made clear in the proof of Claim~\ref{lem:Lprime-is-L-2}.

We show that this results in an instance that is equivalent to the original input instance, in the following sense.
\begin{myclaim}\label{lem:2CNCF-to-rootCSP}
$(L,V)$ is a yes-instance of \rootCSP{2} if and only if $G'$ is $2$-CNCF-colorable.
\end{myclaim}

\begin{claimproof}
Suppose $\tau\colon V \rightarrow \{0,1\}$ is a satisfying assignment for $(L,V)$. We show how to define a $2$-CNCF coloring $c\colon V \rightarrow \{\red,\blue\}$ for $G'$. Let $c(v) := \red$ if $\tau(r_v) = 1$ and let $c(v) := \blue$ if $\tau(b_v) = 1$. Note that this defines exactly one color for each vertex, as by Step \ref{step:equations:red-or-blue}, $r_v + b_v = 1$ and we used at most two distinct colors. It remains to show that this is indeed a CNCF-coloring. Let $v \in V(G')$ be an arbitrary vertex, we show that $N[v]$ is conflict-free colored. It follows from the equations added in Step \ref{step:equations:per-vertex}, that one of the following holds.
\begin{itemize}
\item $(\sum_{u \in N[v]} \tau(r_v) = 1)$. In this case, $N[v]$ contains exactly one vertex $u \in N[v]$ with $c(u) = \red$, showing that $N[v]$ is conflict-free colored.
\item $(\sum_{u \in N[v]} b_v = 1)$. In this case, $N[v]$ contains exactly one vertex $u \in N[v]$ with $c(u) = \blue$, showing that $N[v]$ is conflict-free colored.
\end{itemize}
This concludes this direction of the proof.

For the other direction, suppose $G'$ has $2$-CNCF-coloring $c$, we show how to define a satisfying assignment $\tau$ for $(L,V)$. For $v \in V(G')$, let $\tau(r_v) := 1$ if $c(v) = \red$ and let $\tau(r_v) := 0$ otherwise. Similarly,  $\tau(b_v) := 1$ if $c(v) = \blue$ and let $\tau(r_v) := 0$ otherwise. Observe that by this definition, $\tau(r_v) = 1-\tau(b_v)$ for all $v \in V(G)$, showing that we satisfy all equations introduced in Step~\ref{step:equations:red-or-blue}. For the equations introduced in Step~\ref{step:equations:per-vertex}, consider an arbitrary vertex $v \in V(G')$. Suppose its CNCF-color is \red, then $N[v]$ contains exactly one vertex $u$ with $c(u) = \red$ and thus $\sum_{u \in N[v]} r_u =1$, implying $(-1 +\sum_{u \in N[v]} r_v )\cdot (-1+\sum_{u \in N[v]} b_v - 1) = 0$ as desired. Similarly, if its CNCF-color is \blue we obtain $\sum_{u \in N[v]} b_u =1$ and again $(-1 +\sum_{u \in N[v]} r_v )\cdot (-1+\sum_{u \in N[v]} b_v) = 0$. It remains to prove that the equations added in Step~\ref{step:equations:ensure-nbhood} are satisfied. For this, let $v$ be an arbitrary vertex for which the equation was added. Observe that if $v$ is colored \red, then its neighborhood contains no \red vertices, such that $\sum_{u \in N(v)} r_u = 0$ and the equation is satisfied, or $d-1$ \red vertices, such that $\sum_{u \in N(v)} r_u = d-1$ and again the equation is satisfied. If $v$ is colored \blue, then either its neighborhood is entirely \red, such that $\sum_{u \in N(v)} r_u = d-1$ or it contains exactly one \red vertex, such that  $\sum_{u \in N(v)} r_u = 1$. In both cases the equation is satisfied.
\qed\end{claimproof}

Clearly, $|V| = 2n$ if $n$ is the number of vertices of $G'$.
We will now show how to modify $L$, such that it uses only variables for the vertices in $S \cup S'$.
To this end, we introduce the following function. For $v \notin (S\cup S')$, let $f_v(V) := g\big(\sum\nolimits_{u \in N(v)} r_u, |N(v)|\big)$, where
\[g(x,N) = -\frac{(N-x)(x-1)(N-2(x+1))}{N(N-2)}.\]
Note that for any fixed $N > 2$, $g(x,N)$ describes a degree-$3$ polynomial in $x$ over $\mathbb{Q}$. The following is easy to verify.
\begin{observation}\label{lem:properties-of-g}
$g(0,N) =g(N-1,N)= 1$, and $g(N,N) = g(1,N) = 0$ for all  $N\in \mathbb{Z} \setminus \{0,2\}$.
\end{observation}
Observe that $f_v$ only uses variables defined for vertices that are in $S$. As such, let $V' := \{r_v,b_v \mid v \in S\} \cup \{r_v,b_v \mid v \in S'\}$, and let $L'$ be equal to $L$ with every occurrence of $r_v$ for $v \notin (S\cup S')$ substituted by  $f_v$ and every occurrence of $b_v$ for $v \notin (S\cup S')$ substituted by $(1-f_v(V))$.

\begin{myclaim}\label{lem:Lprime-is-L-1}
If $\tau\colon V \rightarrow \{0,1\}$ is a satisfying assignment for $(L,V)$, then $\tau|_{V'}$ is a satisfying assignment for $(L',V')$.
\end{myclaim}

\begin{claimproof}
We show this by showing that for all $v \notin (S\cup S')$, $f_v(\tau(V)) = \tau(r_v)$ in this case. Since $\tau(b_v) = 1-\tau(r_v)$ by the constraints added in Step~\ref{step:equations:red-or-blue}, this will conclude the proof. Consider an arbitrary vertex $v \notin S$. Observe that by the equations added in Step \ref{step:equations:per-vertex}, we are in one of the following cases.
\begin{itemize}
\item $\sum_{u \in N[v]} \tau(r_u) = 1$ and $\tau(r_v) = 1$. In this case, $\sum_{u \in N(v)} \tau(r_u) = 0$ and thereby $f_v(V) = g(0,|N(v)|) = 1 = \tau(r_v)$, by Observation \ref{lem:properties-of-g}.
\item $\sum_{u \in N[v]} \tau(r_u) = 1$ and $\tau(r_v) = 0$. In this case, $\sum_{u \in N(v)} \tau(r_u) = 1$ and thereby $f_v(V) = g(1,|N(v)|) = 0 = \tau(r_v)$, using Observation \ref{lem:properties-of-g}.
\item $\sum_{u \in N[v]} \tau(b_u) = 1$ and $\tau(b_v) = 1$. In this case, since $\tau(b_v) = 1-\tau(r_v)$ for all $v$, we obtain that $\sum_{u \in N(v)} \tau(b_u) = 0$, and thus $\sum_{u \in N(v)} \tau(r_u) = |N(v)|$. Thereby, $f_v(V) = g(|N(v)|,|N(v)|) = 0 = 1-\tau(b_v) = \tau(r_v)$ by Observation \ref{lem:properties-of-g}.
\item $\sum_{u \in N[v]} \tau(b_u) = 1$ and $\tau(b_v) = 0$. Hereby, $\sum_{u \in N(v)} \tau(b_u) = 1$ and thus $\sum_{u \in N[v]} \tau(r_u) = |N(v)| - 1$. Thus, $f_v(V) = g(|N(v)|-1,|N(v)|) =1=1-\tau(b_v)  = \tau(r_v)$ using Observation \ref{lem:properties-of-g}.\qedhere
\end{itemize}
\end{claimproof}

The next claim shows the equivalence between $(L',V')$ and $(L,V)$.
\begin{myclaim}\label{lem:Lprime-is-L-2}
If $\tau \colon V' \rightarrow \{0,1\}$ is a satisfying assignment for $(L',V')$, then there exists a satisfying assignment $\tau'\colon V\rightarrow \{0,1\}$ for $(L,V)$ such that $\tau'|_{V'} = \tau$.
\end{myclaim}

\begin{claimproof}
Let $\tau$ be given, we show how to construct $\tau'$. For all $x \in V'$, let $\tau'(x) := \tau(x)$. Furthermore, for $r_v \in V\setminus V'$, let $\tau(r_v) := f_v(\tau(V))$ and let $\tau(b_v) := 1 - \tau(r_v)$. Since $L'$ was simply obtained from $L$ by substituting $r_v$ by $f_v(V)$ and $b_v$ by $(1-f_v(V))$ in all constraints, it is clear that $\tau'$ satisfied all equations in $L'$. It remains to show that $\tau(r_v) \in \{0,1\}$ for all $v \in V$. If $v \in V'$, this is obvious, so suppose $v \in V \setminus V'$ such that $\tau(r_v) = f_v(\tau(v))$. Observe that an equation was added for $v$ in Step~\ref{step:equations:ensure-nbhood}. Therefore, we know that $\sum_{u \in N(v)} r_u \in \{0,1,d_v-1,d_v\}$ and it follows from Lemma~\ref{lem:properties-of-g} that $f_v(\tau(V))$ takes a boolean value. \qed
\end{claimproof}

Using the method described above, we obtain an instance $(L,V)$ of \rootCSP{2} such that $(L,V)$ has a satisfying assignment if and only if $G$ is $2$-CNCF-colorable by Claims~\ref{claim:kernel:CNCF:preprocess} and~\ref{lem:2CNCF-to-rootCSP}. Then we obtain an instance $(L',V')$ such that $(L',V')$ is satisfiable if and only if $(L,V)$ is satisfiable by Claims \ref{lem:Lprime-is-L-1} and \ref{lem:Lprime-is-L-2}. As such, $(L',V')$ is a yes-instance if and only if $G$ is $2$-CNCF-colorable and it suffices to give a kernel for $(L',V')$. Observe that $|V'| = \Oh(k^2)$.

We start by partitioning $L'$ into three sets $L'_S$, $L'_1$ and $L'_2$. Let $L'_S$ contain all equalities created for a vertex $v \in S$. Let $L'_1$ contain all equations that contain at least one of the variables in $\{r_v,b_v \mid v \in S'\}$ and let $L_2$ contain the remaining equalities. Observe that $|L'_S| = k$ by definition. Furthermore, the polynomials in $L'_1$ have degree at most $2$, as they were created for vertices in $V(G') \setminus S$, and these are not connected. As such, we use Theorem~\ref{thm:subset_kernel_Q} to obtain $L''_1 \subseteq L'_1$ such that $|L''_1| = \Oh((k^2)^2) = \Oh(k^4)$ and any boolean assignment satisfying all equalities in $L''_1$ satisfies all equalities in $L'_1$.

Similarly, we observe that $L'_2$ by definition contains none of the variables in $\{r_v,b_v \mid v \in S'\}$, implying that the equations in $L'_2$ are equations over only $k$ variables. Since the polynomials in $L'_2$ have degree at most $6$, we can apply Theorem~\ref{thm:subset_kernel_Q} to obtain $L''_2 \subseteq L'_2$ such that $|L''_2| \leq \Oh(k^6)$ and any assignment satisfying all equations in $L''_2$ satisfies all equalities in $L'_2$.

We now define $L'' := L''_1 \cup L''_2 \cup L'_S$, and the output of our polynomial generalized kernel will be $(L'',V')$.
The correctness of the procedure is proven above, it remains to bound the number of bits needed to store instance $(L'',V')$.

By this definition, $|L''| \leq \Oh(k^{6})$. To represent a single constraint, it is sufficient to store the coefficients for each variable in $V'$. The storage space needed for a single coefficient is $\Oh(\text{log}(n))$, as the coefficients are bounded by a polynomial in $n$. Thereby, $(L'',V')$ can be stored in $\Oh(k^{6}\cdot k^2\log{n})$ bits. To bound this in terms of $k$, we observe that it is easy to solve $2$-\CNCF in time $\Oh(2^{k^2}\cdot \text{poly(n)})$. This is done by guessing the coloring of $S$, extending this coloring to the entire graph (observe $G \setminus S$ has no vertices of degree less than three) and verifying whether this results in a CNCF-coloring. Therefore, we can assume that $\log(n) \leq k^2$, as otherwise we can solve the $2$-\CNCF problem in $\Oh(2^{k^2}\text{poly}(n))$ time, which is then polynomial in $n$. Thereby we conclude that
$(L'',V')$ can be stored in $\Oh(k^{10})$ bits.
\qed\end{proof}

\subsection{Kernelization bounds for conflict-free coloring extension}
\label{sec:extension-main}
We furthermore provide kernelization bounds for the following extension problems.

\defproblem{\textsc{$q$-\CNCF-VC-Extension}}{A graph $G$ with vertex cover $S$ and partial $q$-coloring $c\colon S \rightarrow [q]$.}{Does there exist a $q$-CNCF-coloring of $G$ that extends $c$?}

\noindent We define \textsc{$q$-\ONCF-VC-Extension} analogously.

We obtain the following kernelization results when parameterized by vertex cover size, thereby classifying the situations where the extension problem has a polynomial kernel. The extension problem turns out to have a polynomial kernel in the same case as the normal problem. However, we manage to give a significantly smaller kernel.  Observe that the kernelization result is non-trivial, since \textsc{$2$-\CNCF-VC-Extension} is NP-hard (see Theorem \ref{thm:CV-extension-NP} below).

\begin{theorem}\label{thm:extension-kernel-bounds}
The following results hold.
\begin{enumerate}
\item \textsc{$2$-\CNCF-VC-Extension} has a kernel with $\Oh(k^2)$ vertices and edges that can be stored in $\Oh(k^2 \log k)$ bits. Here $k$ is the size of the input vertex cover $S$.
\item $q$-\textsc{CNCF-Coloring-VC-Extension} for any $q \geq 3$, and \textsc{$2$-\ONCF-VC-Extension} parameterized by the size of a vertex cover do not have a polynomial kernel, unless \containment.
\end{enumerate}
\end{theorem}

We start by noting that the kernelization lower bounds given in the previous sections still apply. In particular, we obtain the following two Corollaries.
\begin{corollary}
$2$-\textsc{ONCF-Coloring-VC-Extension} parameterized by the size of a vertex cover does not have a polynomial kernel, unless \containment.
\end{corollary}
\begin{proof}
Observe that in the cross-composition given in Lemma~\ref{lem:ONCF-kernel-LB}, the vertices
\begin{align*}
H &\cup \{R,R',B,B',a,a_1,a_2,a_1',a_2'\}\cup \{v_j,v_j^1,v'^1_j,v''^1_j,v_j^2,v'^2_j,v''^2_j\mid j \in [n]\}\\&\cup S \cup  \{\text{all vertices in gadgets labeled }g_{10}, g_9, g_8,g_4, \text{ or } g_5\}
\end{align*}
form a vertex cover of appropriately bounded size, and that these vertices receive always receive the same color in the proof of Claim~\ref{claim:CNCF:if}.
\qed\end{proof}

Furthermore, Lemma~\ref{lem:3-CNCF-kernel-LB} immediately gives us the following result on this extension problem.
\begin{corollary}
For any $q \geq 3$, $q$-\textsc{CNCF-Coloring-VC-Extension} parameterized by the size of a vertex cover does not have a polynomial kernel, unless \containment.
\end{corollary}
\begin{proof}
The result follows immediately from the same cross-composition as given in the proof of Lemma \ref{lem:3-CNCF-kernel-LB}. Observe that the vertices $C \cup H \cup \{v_j \mid j \in [n]\} \cup \{s^{i,i'}_{j,j'} \mid i,i'\in[k], j,j'\in [n]\} \cup \{a,a_1,a_1',a_2,a_2',a_3\}$ form a vertex cover of the created graph $G$ of size $\text{poly}(n)$, and that they are always given the same coloring in the proof of Claim \ref{claim:CNCF:if}.\qed
\end{proof}

The results above prove part 2 of Theorem \ref{thm:extension-kernel-bounds}.
 We will now show that \textsc{$2$-\CNCF-VC-Extension} has a simple polynomial kernel of size  $\Oh(k^2\log{k})$, where $k$ is the size of the vertex cover. This proves part 1 of Theorem ~\ref{thm:extension-kernel-bounds}. We start by arguing that \textsc{$2$-\CNCF-VC-Extension} is indeed NP-hard.

\begin{theorem}\label{thm:CV-extension-NP}
$2$-\textsc{CNCF-Coloring-VC-Extension} is NP-hard.
\end{theorem}
\begin{proof}
We prove this by a reduction from \textsc{Monotone Exact Sat}, which is defined as follows.

\defproblem{\textsc{Monotone Exact Sat}}{A formula $\mathcal{F}$ over variable set $X$ that is a conjunction of clauses, where each clause consists of a number of variables from $X$. }{Does there exist an assignment $\tau\colon X \rightarrow \{0,1\}$ such that every clause in $\mathcal{F}$ contains exactly one variable that is set to $1$?}

It is known that the \textsc{Monotone Exact Sat} problem is NP-hard, as it generalizes problem \textsc{NP1} in \cite{SchaeferIntractability}. Let an instance $\mathcal{F} = C_1 \wedge \dots \wedge C_m$ over variables $X = \{x_1,\ldots,x_n\}$ be given, we show how to construct a graph $G$ with vertex cover $S$ and precoloring $f \colon S \rightarrow \{\red,\blue\}$ for $2$-\textsc{CNCF-Coloring-VC-Extension}. See Figure \ref{fig:2-CNCF-extension-NP} for a sketch of $G$.
\begin{figure}
\centering
\includegraphics{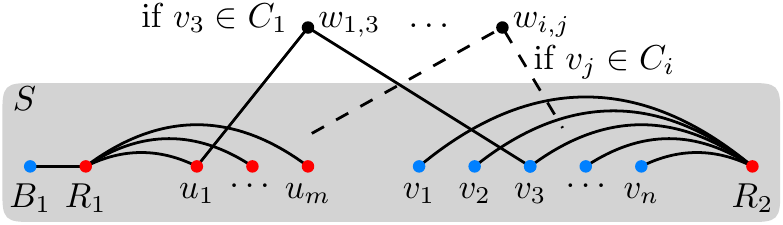}
\caption{The reduction from  \textsc{Monotone Exact Sat} to \textsc{$2$-\CNCF-VC-Extension}.}
\label{fig:2-CNCF-extension-NP}
\end{figure}
\begin{itemize}
\item Add vertices $R_1$, $R_2$, and $B_1$ to $G$ and to $S$. Let $c(R_1):= c(R_2) := \red$ and let $c(B_1) := \blue$. Connect $R_1$ to $B_1$.
\item For every clause $C_i$ in $\mathcal{F}$, add a vertex $u_i$ to $G$. Let $u_i$ be contained in $S$ and let $c(u_i) := \red$. Connect all vertices $u_i$ to $R_2$.
\item For every variable $x_j \in X$, add a vertex $v_j$ to $G$, let $v_j$ be contained in $S$ and set $c(v_j) := \blue$. Connect $v_j$ to $R_2$.
\item For every $i \in [m]$, $j \in [n]$, if variable $x_j$ is in clause $C_i$, construct a new vertex $w_{i,j}$ and connect $w_{i,j}$ to $v_j$ and $u_i$. $w_{i,j}$ is \emph{not} contained in $S$.
\end{itemize}
Clearly, $S$ is hereby a vertex cover of $G$ and it is colored by $c$. It is easy to see that the construction above can be done in polynomial time. It remains to show that $c$ can be extended to a $2$-CNCF-coloring of $G$ if and only if $\mathcal{F}$ is satisfiable.

$(\Rightarrow)$ Suppose $\mathcal{F}$ is satisfiable and has satisfying assignment $\tau\colon X\to \{0,1\}$ , we give a $2$-CNCF-coloring $c' \colon V(G) \to \{\red,\blue\}$ of $G$, such that $c'$ extends $c$. Naturally, for every vertex $s \in S$, let $c'(s) := c(s)$. For $i\in [m]$, $j \in [n]$ let $c(w_{i,j}) := \blue $ if $\tau(x_j) = 1$ and let $c(w_{i,j}) := \red$ otherwise. It remains to show that this is a valid CNCF-coloring of $G$. We check the neighborhoods of all vertices.

$N[R_1] = \{R_1,B_1\} \cup \{u_i \mid i \in [m]\}$. Here $B_1$ is a unique \blue vertex. $N[R_2] = \{R_2\} \cup \{v_j \mid j \in [n]\}$ and $R_2$ is a unique \red vertex in this set. $N[B_1] := \{B_1,R_1\}$ and these vertices are \red and \blue, as desired. For any vertex $u_i$, $N[u_i]$ contains one \blue vertex, namely $w_{i,j}$ where $j$ is such that $x_j$ is the \emph{unique} variable in clause $C_i$ with $\tau(x_j) = 1$.
For any $j \in [n]$, $N[v_j] := \{w_{i,j} \mid x_j \in C_i\} \cup \{v_j, R_2\}$. Since all vertices in $\{w_{i,j} \mid x_j \in C_i\}$ receive the same color, this set has a uniquely colored vertex which is either $v_j$ (which is \blue) or $R_2$ (which is \red). For any $i \in [m], j\in [n]$, vertex $w_{i,j}$ has exactly two neighbors and these receive different colors, and thus $N[w_{i,j}]$ is has a neighbor with a unique color.

$(\Leftarrow)$ Let $c' \colon V(G) \to \{\red,\blue\}$ be a CNCF-Coloring of $G$.  Let $j \in [n]$, then $N[v_j] = \{w_{i,j} \mid v_j \in C_i\} \cup \{R_2,v_j\}$. Thereby, we observe that the vertices in $W_j := \{w_{i,j} \mid v_j \in C_i\}$ all receive the same color since $c'(v_j) = \blue$ and $c'(R_2) = \red$. Let $\tau(x_j) := 1$ if all vertices in $W_j$ have color \blue, and let $\tau(x_j) := 0$ otherwise. We show that $\tau$ is a satisfying assignment for $\mathcal{F}$. Let $C_i$ be a clause of $\mathcal{F}$, we show that there is exactly one $x_j$ such that $\tau(x_j) = 1$, by showing that there is exactly one $j \in [m]$ such that vertex $c'(w_{i,j}) = \blue$. Consider vertex $u_i$, then $N[u_i] = \{u_i, R_1\}\cup \{w_{i,j} \mid v_j \in C_i\}$. Since $c'(u_i) = c'(R_1) = \red$ since $c'$ extends $c$, it trivially follows that there is indeed a unique $j$ such that $c'(w_{i,j}) = \blue$. This concludes the proof.\qed
\end{proof}

We now show that, unlike $3$-\textsc{CNCF-Coloring-VC-Extension}, $2$-\textsc{CNCF-Coloring-VC-Extension} has a simple polynomial kernel.

\begin{lemma}
$2$-\textsc{CNCF-Coloring-VC-Extension} parameterized by the size of the vertex cover has a kernel of size $\Oh(k^2\log k)$. \end{lemma}
\begin{proof}
Let $G$ with partial coloring $c$ and vertex cover $S$ be an instance of the problem, with $|S| \leq k$. We first show how to obtain an equivalent instance $(G',S',c')$ such that $|S'| \leq 3|S|$ and such that every vertex in $V(G')\setminus S'$ has degree at most two. Then we can use a procedure given by Gargano and
Rescigno \cite{GarganoRR15} to further reduce the number of vertices in $V(G')\setminus S'$ to at most $\Oh(|S'|^2) = \Oh(k^2)$.

If there exists a vertex that has at least two \red and two \blue neighbors by this precoloring, output a trivial no-instance. For the rest of the kernelization, we can thus assume that this case does not occur.
Initialize $(G',S',c')$ as $(G,S,c)$. Observe that for any vertex of degree at least $3$, its coloring is now completely determined. While there exists a vertex $v$ of degree at least three in $G'\setminus S'$, we define $c'(v)$ as follows. If $v$ has only \red neighbors, let $c'(v) := \blue$. Otherwise, if $v$ has at least two \red and exactly one \blue neighbor, let $c'(v) := \red$. Similarly, if $v$ has only \blue neighbors let $c'(v):=\red$ and if $v$ has exactly one red neighbor let $c'(v):=\blue$. It is easy to see that there is a $2$-CNCF coloring of $G$ that extends $c$, if and only if there is one extending $c'$.

For each vertex $u \in S'$, mark two neighbors that are colored \red by $c'$, and two that are \blue (if these exist). Add all marked vertices to $S'$, and delete all vertices $v \in G \setminus S'$ that have degree at least three and are not marked. Observe that hereby, $c'$ is a coloring of the vertices of $S'$, $S'$ is a vertex cover of $G'$, and each vertex in $V(G')\setminus S'$ has degree at most two. We argue the following.

\begin{myclaim}
The graph $G'$ has a $2$-CNCF-coloring extension if and only if $G$ has a $2$-CNCF-coloring extension.
\end{myclaim}
\begin{claimproof}
Suppose $G$ has a $2$-CNCF-coloring extension of $c$, by the observation above there is also a $2$-CNCF-extension of $c'$, let this be $c''$. We show that $c''|_{V(G')}$ is a proper coloring of $G'$. Every vertex in $G' - S$ has the same neighborhood as in $G$, and thus this neighborhood is conflict-free colored by $c''$. For every vertex in $s \in S$, $N[s] \cap (S \cup \{v \notin S \mid d(v) \leq 2\})$ is the same in $G$ and $G'$. For the vertices in $S \cup \{v \notin S \mid d(v) = 2\})$ the color is the same for any $2$-CNCF-coloring of $G$ and we kept two \red and two \blue vertices. As such, $c''$ is a CNCF-coloring of $G'$.

Suppose $G'$ has a $2$-CNCF-coloring extension $c''$ of $c'$.
We define a $2$-CNCF-coloring $d$ of $G$ that extends $c$. Start by defining $d(v) = c''(v)$ for any vertex $v \in V(G) \cap V(G')$. Hereby, all vertices in the vertex cover $S$ of $G$ are colored. Let $v \in V(G) \setminus S$. Note that $v$ has at least three neighbors in $s$, as otherwise $v$ would have been a vertex in $G'$. Note that $N(v) \subseteq S$. Define $d(v)$ as \red if $N(v)$ has only \blue vertices. Furthermore, let $d(v) := \red$ if $N(v)$ contains exactly one \blue vertex. In all other cases, define $d(v):= \blue$. This concludes the definition of $d$, it remains to show that $d$ is indeed a CNCF-coloring.

Clearly, by this definition, for any $v \notin S$ we have that $N(v)$ is conflict-free colored by $d$, as we assumed that no such vertex had two \red and two \blue neighbors. It remains to show that for $v \in S$, $N(v)$ is conflict-free colored. Suppose for contradiction that it is not. Since any vertex $v \in S$ was conflict-free colored by $c''$ in $G'$, this implies that there exists a vertex $v \in S$ that has two \red and two \blue neighbors under $d$. Without loss of generality, suppose \red was the conflict-free color of $v$ in $G'$. Thus, there is a vertex $w \in V(G)\setminus V(G')$ that is a neighbor of $v$, with $d(w) := \red$. But this contradicts that $w$ is removed by the marking procedure, as we always keep at least two \red neighbors of $v$ if they exist. Thereby, $d$ is a CNCF-coloring of $G$.
\qed\end{claimproof}

To obtain the kernel, for every set $X \subseteq S'$ of size at most two, mark $3$ vertices $v \in V(G')$ with $X  = N(v)$, if less than three such vertices exist, mark all. Remove all unmarked vertices from $V(G')\setminus S'$. This concludes the procedure. It follows from \cite[Lemma 6]{GarganoRR15} that this last step does not change the $2$-CNCF-colorability of $G'$, observe that this still holds after predefining the coloring of the vertex cover.
It is easy to observe that $|S'|\leq 3|S|$ and $|V(G')| \leq |S'|^2 = \Oh(k^2)$. Furthermore, since any vertex in $G' \setminus S'$ has degree at most two, $|E(G')| \leq |S'|^2 + 2|V(G')\setminus S'| = \Oh(k^2)$. Using adjacency lists, this kernel can thus be stored in $\Oh(k^2 \log k)$ bits.
\qed\end{proof}

This completes the proof of Theorem ~\ref{thm:extension-kernel-bounds}.

\section{Combinatorial bounds}\label{sec:cb}
Given a graph $G$, it is easy to prove that $\chi_{\sf CN}(G) \leq \chi (G)$. However, there are examples that negate the existence of such bounds with respect to $\chi_{\sf ON}$~\cite{GarganoRR15}. In this section, we prove combinatorial bounds for $\chi_{\sf ON}$ with respect to common graph parameters like treewidth, feedback vertex set and vertex cover.

First, note that if $G$ is a graph with isolated vertices then the graph can have no ONCF-coloring. Therefore, in all the arguments below we assume that $G$ does not have any isolated vertices. We obtain the following result. Recall that for a graph $G$, ${\sf vc}(G)$, ${\sf fvs}(G)$ and ${\sf tw}(G)$ denote the size of a minimum vertex cover, the size of a minimum feedback vertex set and the treewidth of $G$, respectively.

\begin{theorem}\label{thm:comb-bds}
Given a connected graph $G$,
\begin{enumerate}
\item $\chi_{\sf ON}(G) \leq 2{\sf tw}(G)+1$,
\item $\chi_{\sf ON}(G) \leq {\sf fvs}(G)+3$,
\item $\chi_{\sf ON}(G) \leq {\sf vc}(G)+1$. Furthermore, if $G$ is not a star graph or an edge-star graph, then $\chi_{\sf ON}(G) \leq {\sf vc}(G)$.
\end{enumerate}
\end{theorem}

In order to prove the above theorem, we prove each item separately. In the following lemma, we consider the bound on $\chi_{\sf ON}$ with respect to the treewidth of a given graph.
\begin{lemma}\label{lem:treewidth-bd}
If $G$ is a graph with treewidth $t$, then $\chi_{\sf ON}(G) \leq 2t+1$.
\end{lemma}
\begin{proof}
Consider a nice tree decomposition ${\mathcal{T}}=(T,\{X_\bu\}_{\bu \in V(T)}))$ of $G$. We give a vertex coloring $c \colon V(G) \rightarrow [2t+1]$ of $G$, which we will prove to be a $(2t+1)$-ONCF-coloring of $G$. Furthermore, we give a function $f \colon V(G) \rightarrow [2t+1]$ such that $f(v)$ is the color that is uniquely used in the neighborhood of $v$. We will color the graph such that
\begin{itemize}
\item If two vertices are in the same bag $X_{\bi}$ of $\mathcal{T}$, they receive distinct colors.
\item The graph induced by the colored vertices is ONCF-colored, with the exception of vertices that are isolated.
\end{itemize}
 Observe that by these two properties, $c$ is not only a ONCF-coloring, but also a proper coloring of $G$.

Let an arbitrary leaf $\br \in T$ be the root, and let the corresponding bag be denoted as $X_\br$. Note that $|X_\br| \leq t+1$. Color each of the vertices in $X_\br$ with a unique color from $[t+1]$. For each vertex $v \in X_\br$ , if $N(v) \cap X_\br \neq \emptyset$, pick one arbitrary vertex $u \in N(v) \cap X_\br$ and let $f(v) := c(u)$. Else, let $f(v) := 0$.

Now, let $\bi$ be a vertex of $T$ such that its parent $\bj$ has been handled. We show how to color the vertices of $X_\bi$. If $X_\bi \subseteq X_\bj$, all vertices of $X_\bi$ have already been taken care of. Given that $T$ is a nice tree decomposition, the only alternative is $X_\bj$ is a forget node in $T$. Hence, $X_\bi = X_\bj \cup \{v\}$ for some $v \in V(G)$. Let $C := \{c(v) \mid v \in X_\bj \}$ and let $F := \{f(v) \mid v \in X_\bj\}$. Since $|X_\bi| \leq t+1$, it follows that $|X_\bj| \leq t$. Thereby, $|F \cup C| \leq 2t$. We color $v$ with a color from $[2t+1]\setminus (F \cup C)$. If there is a vertex $u \in X_\bj$ for which $f(u) = 0$ and $\{u,v\} \in E(G)$, we let $f(u) = c(v)$. Furthermore, if $v$ is not isolated in the graph constructed thus far, we let $f(v) = c(u)$ for some $u \in X_\bi$. Note that in the graph colored so far, $N(v) \subseteq X_\bj$.

We argue that $c$ is a $(2t+1)$-ONCF-coloring of $G$. Consider a vertex $v \in V(G)$. Let $X_\bj$ be the vertex that is a forget node for $v$ and that is closest to $X_\br$. By definition of a nice tree decomposition rooted at $\br$, $X_\bj$ is unique for the vertex $v$. Let $X_\bi$ be the child of $X_\bj$ - by definition $X_\bj$ contains $v$. Let $X_\ba$ be a bag closest to $X_\br$ that contains both $v$ and a neighbor of $v$. Then, by the definition of our coloring $c$, there is a vertex $w \in N(v) \cap X_\ba$ that witnesses the conflict-free coloring of $N(v)$.

Thus, $\chi_{\sf ON}(G) \leq 2t+1$.
\qed\end{proof}

A larger parameter than the treewidth of a graph is the size of a minimum feedback vertex set of a graph. In the following result, we compare $\chi_{\sf ON}$ with the size of a minimum feedback vertex set.
\begin{lemma}\label{lem:fvs-bd}
If $G$ is a graph with a feedback vertex set of size $\ell$, then $\chi_{\sf ON}(G) \leq \ell+3$.
\end{lemma}
\begin{proof}
Let $X$ be a FVS of size $\ell$ in $G$. Using the color set $\mathcal{C} := \{r,g,b\} \cup \{c_x \mid x \in X\}$, we define a vertex coloring $c:V(G) \rightarrow \mathcal{C}$, which we prove is a $(\ell+3)$-ONCF-coloring of $G$.
The general idea is as follows. Vertices in $X$ are colored with $c_x$, and vertices not in $X$ are conflict-free colored with $\{g,b\}$. The only problem with this coloring are vertices in $x \in X$ with no neighbors in $X$. We resolve this issue by a careful recoloring of some of these vertices with color $r$, and some of the vertices in $V(G)\setminus X$ with colors from $\{c_x \mid x \in X\}$.

Let $I := \{x \in X \mid N(x) \cap X = \emptyset\}$. Initialize $L = I$. Color the vertices in $G$ by the following procedure.
\begin{enumerate}
\item For all $x \in X \setminus L$, let $c(x) := c_x$.
\item\label{line:2} Each connected component $Y$ of $G-X$ is a tree. If $\vert Y \vert >1$, then ONCF-color $Y$ with colors $g$ and $b$. Otherwise, (when $|Y| = 1$), color the single vertex in $Y$ with $g$.
\item \textbf{while} there exists $u \in V(G) \setminus X$ such that $N(u) \subseteq L$
\item \quad For each $x \in N(u)$ let $c(x) := c_x$
\item \quad Let $L := L \setminus N(u)$
\item \quad Pick an arbitrary $x\in N(u)$ and let $c(u) := c_x$
\item\label{line:7} \textbf{for each} $x \in L$
\item \quad Let $c(x) := r$
\item\label{line:9} \quad If there is no $y \in L$, $v \in N(x)$ such that $c(v) = c_y$
\item \qquad Pick an arbitrary $v \in N(x)$ and recolor $v$ by letting $c(v) := c_x$
\end{enumerate}

We argue that $c$ is a ONCF-coloring. We consider the different vertices in $G$ and argue their ONCF-coloring by $c$. Let $L$ be the set $L$ as in line \ref{line:7} of the procedure.
\begin{itemize}
 \item Consider a vertex $v \in V(G) \setminus X$, if $N(v) \setminus X \neq \emptyset$. In line \ref{line:2},  $N(v) \setminus X$ is given a conflict-free coloring using the colors $\{b,g\}$, which are not used for vertices in $X$. In line \ref{line:9}, some of these vertices may be recolored to $c_x $ for $ x \in L$, but observe that any such color is used at most once in $G$. This implies that $N(v)$ is conflict-free colored by $c$.
 \item For a vertex $v \in V(G) \setminus X$, if $N(v) \subseteq X$, then it cannot be the case that $N(v) \subseteq \{x \in X \mid c(x) =r \}$ as this would contradict the while loop in the definition of the coloring. All other colors are used at most once on vertices in $X$ and therefore $N(v)$ is ONCF-colored by $c$.
 \item For a vertex $x \in X \setminus I$, there is a neighbor in $X$ the color of which is unique in $V(G)$. Hence, $N(x)$ is ONCF-colored by $c$.
 \item Consider a vertex $x \in I$ with $c(x) = c_x$. Note that $N(x) \subseteq V(G) \setminus X$. By definition of $c$, there is a vertex $v \in V(G) \setminus X$ and $y \in I$ such that $x,y \in N(v) \subseteq I$ and $c(v) = c_y$. The color $c_y$ is used at most once in $N(x)$, and thus it is ONCF-colored by $c$.
 \item Let $x \in I$ be a vertex with $c(x) = r$. Note that $N(x) \subseteq V(G) \setminus X$. By definition of $c$, there is a vertex $w \in N(x)$ that obtains a color $c_y$, for some $y \in I$. Since color $c_y$ is used only once in $V(G)\setminus X$, vertex $w$ witnesses the ONCF-coloring of $N(x)$.
\end{itemize}
Thus, $\chi_{\sf ON}(G) \leq \ell+3$.
\qed\end{proof}
Observe that the bound given in Lemma \ref{lem:fvs-bd} is close to being tight. If we start from a clique $K_k$ and subdivide each edge, the resulting graph has a feedback vertex set of size $k-2$ and needs $k$ colors to be ONCF-colored.

The next lemma bounds the value of $\chi_{\sf ON}(G)$ for graphs with a vertex cover of size~$k$. In particular, we improve the bound given by Gargano and Rescigno \cite[Lemma 4]{GarganoRR15}, who showed that $\chi_{\sf ON}(G) \leq 2k+1$.

\begin{lemma}\label{lem:vc-bd}
Let $G$ be a connected graph with ${\sf vc}(G) = k$. Then $\chi_{\sf ON}(G) \leq k+1$. Furthermore, if $G$ is not a star graph or an edge-star graph, then $\chi_{\sf ON}(G) \leq k$.
\end{lemma}
\begin{proof}
See Figure \ref{fig:vc-coloring} for a sketch of the colorings described in the proof.
\begin{figure}[t]
\centering
\includegraphics{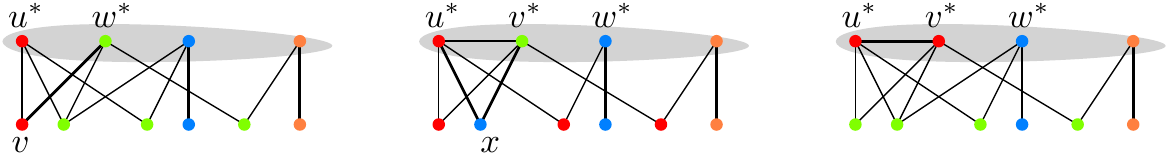}
\caption{(left) A coloring of the graph when all vertices in $G[S]$ are isolated. (middle) The case where $G[S]$ contains an edge and the endpoints have a common neighbor. (right) The case where $G[S]$ contains an edge and the endpoints have no common neighbors.}
\label{fig:vc-coloring}
\end{figure}

We start by proving the bounds for the case where $G$ is not a star and not an edge-star. Let $S$ be a minimum vertex cover of $G$ and let $k$ be the size of $S$. We do a case distinction on the size and connectedness of $S$.

\textbf{($k=2$ and $S$ connected)}
First, we prove the bounds for $k=2$ and $G[S]$ is an edge $\{u^*,v^*\}$. Note that $G$ is not an edge-star graph. Therefore at least one of $u^*$ or $v^*$ have neighbors with degree exactly $1$ in $G \setminus S$.  We show that it is possible to ONCF-color such a graph with $2$ colors, namely $r$ and $b$.
Without loss of generality, let $u^*$ have degree-$1$ neighbor $w^*$. We proceed as follows. Let $c(u^*) := c(w^*):= r$, and $c(v^*) := b$. For any other vertex in $V(G)$, let $c(v) := b$. It remains to verify that this is a ONCF-coloring. Any $v \notin S$ is clearly ONCF-colored by the fact that their neighborhood is a subset of $S$, and the vertices in $S$ receive different colors. Furthermore, $v^*$ has exactly one neighbor of color $r$ (namely $u^*$), and $u^*$ has one neighbor of color $r$, namely $w^*$, concluding this part of the proof.

\textbf{($G[S]$ disconnected or $k \geq 3$)}
We now prove the bounds for $k = 2$  and $G[S]$ is disconnected, and $k \geq 3$. We consider a number of cases.

\emph{(Suppose $G[S]$ contains a connected component $C$ of size at least three.)} Let $v^* \in C$ be a vertex such that $G[C\setminus \{v\}]$ remains connected. We color the vertices in $G$ as follows. For every vertex $u \in S$, let $c(u) := c_u$. For every vertex $u \in S$ that is isolated in $G[S]$, pick an arbitrary neighbor $v \notin S$ and (re)color $v$ such that $c(v) := c_u$. Notice that a vertex $v$ in $G\setminus S$ may be picked multiple times as the candidate for an arbitrary neighbor for an isolated vertex in $S$, and in this case the color of this vertex $v$ is set to the last color it is assigned.
For every vertex $v$ that is not yet colored, let $c(v) := c_{v^*}$.

 Note that by this definition, every vertex in $S$ has a distinct color. The colors that appear on vertices of $G\setminus S$ are either $c_{v^*}$ or the color of a vertex that is isolated in $G[S]$. Also, by the choice of $v^*$, every vertex in the component $C$ of $G[S]$ has at least one other neighbor in $C$. We verify that $c$ is an ONCF-coloring.
\begin{itemize}
\item From the above, every vertex $u \in S$ that belongs to a connected component of size at least two in $G[S]$ has a uniquely colored neighbor in S that witnesses the ONCF-coloring of $N(u)$. Note that this includes~$v^*$.
\item Every isolated vertex $u$ in $G[S]$ has $N(u) \subseteq V(G) \setminus S$. By description, every color $c_v$ for $v \neq v^*$ occurs at most once in $g \setminus S$, and $u$ sees at least one such vertex. Thus, $N(u)$ is ONCF-colored.
\item Finally, since all vertices in $S$ are distinctly colored, the neighborhoods of vertices in $V(G) \setminus S$ are ONCF-colored.
\end{itemize}
Also, notice that the number of colors used is $k$. Thus, we are done in this case.

\emph{(Suppose $G[S]$ only contains connected components of size one.)} Note that $|S| > 1$. Start by letting $c(v) := c_v$ for every vertex $v \in S$. Since $G$ is connected, there exists $v \notin S$ such that $|N(v)| \geq 2$. Pick two vertices $u^*,w^* \in N(v)$ with $u^* \neq w^*$. Let $c(v):= c_{u^*}$. For every vertex $u \in S\setminus \{u^*,w^*\}$ pick an arbitrary neighbor $v \notin S$ and recolor $v$ to $c_u$. Color the vertices that remained uncolored by this procedure with $c_{w^*}$. Note that by this procedure, every vertex in $S$ has a distinct color. To see that $c$ is a ONCF-coloring:
\begin{itemize}
\item Every vertex $v$ in $G[S]$ has $N(v) \subseteq V(G) \setminus S$. By description, every color $c_u$ for $u \neq w^*$ occurs at most once in $G \setminus S$. Furthermore, $v$ has at least one neighbor with a color in $\{c_u \mid u\neq w^* \wedge u \in S\}$. Thus, $N(v)$ is ONCF-colored.
\item Finally, since all vertices in $S$ are distinctly colored, the neighborhoods of vertices in $V(G) \setminus S$ are ONCF-colored.
\end{itemize}
Also, notice that the number of colors used is $k$. Thus, we are done in this case.

\emph{(Otherwise.)} In this case $G[S]$ has size at least $3$, contains multiple connected components, and at least one such component has size two. Let $C = \{u^*,v^*\}$ be a connected component in $G[S]$ and let $w^*$ be another arbitrarily chosen vertex. We do a further case distinction.
\begin{itemize}
\item Suppose there exists a vertex $x \notin S$ with $N(x) = \{u^*,v^*\}$, take one arbitrary such vertex. Then we let $c(v) := c_v$ for all $v \in S$, and we let $c(x) := c_{w^*}$. For every vertex $u \in S$ that is isolated in $G[S]$, pick an arbitrary neighbor $v \notin S$ and recolor $v$ to $c(v) := c_u$.
    Define $c(v) := c_{u^*}$ for all vertices $v$ that remained uncolored thus far. Notice that this coloring is very similar to the colorings in the previous cases. The only verification to be done is that for the sets $N(u^*)$ and $N(v^*)$ and both these neighborhoods have a vertex $x$ that is uniquely colored with $c_{w^*}$. Thus, with arguments similar to those in the previous cases we obtain a $k$-ONCF-coloring for $G$ in this case.
\item Alternatively, if there exists no vertex $x$ with $N(x) = \{u^*, v^*\}$, then we let $c(u^*) = c(v^*) = c_{u^*}$, and we let $c(v) := c_v$ for all vertices in $S \setminus \{u^*,v^*\}$. For every vertex $u \in S$ that is isolated in $G[S]$, pick an arbitrary neighbor $v \in N(u)$ and recolor $v$ to $c(v) := c_u$. Color all remaining vertices with $c_{v^*}$. Notice that $N(u^*)$ and $N(v^*)$ have $v^*$ and $u^*$ uniquely colored with $c_{u^*}$, respectively. Using arguments similar to previous cases, we can show that the described coloring is a $k$-ONCF-coloring.
    \end{itemize}

If $G$ is not a star and not an edge-star, we are in one of the cases above. Otherwise, it is easy to observe that stars have a vertex cover of size one and can always be colored with two colors, and edge-stars can be colored with three colors while having a minimum vertex cover size of two.
\qed\end{proof}

Observe that the bounds of Lemma~\ref{lem:vc-bd} are tight. First, a star graph requires~$2$ colors and has vertex cover size~$1$ while an edge-star graph requires~$3$ colors and has vertex cover size $2$. On the other hand, given an $q \geq 3$, taking the complete graph $K_q$ and subdividing each edge once results in a graph that requires~$q$ colors~\cite{GarganoRR15} for an ONCF-coloring and has a vertex cover of size~$q$.

Using Lemmas~\ref{lem:treewidth-bd}, \ref{lem:fvs-bd} and \ref{lem:vc-bd} we complete the proof of Theorem~\ref{thm:comb-bds}. 
 \section{Open Problems}\label{sec:conclusion}
The study in this paper leads to some interesting open questions. In this paper we only exhibit a generalized kernel of size $\Oh(k^{10})$ for $2$-\CNCF and it remains to resolve the size of tight polynomial kernels for the problem. 
 On the combinatorial side, with respect to minimum vertex cover, we obtain tight upper bounds on $\chi_{ON}(G)$. It would be interesting to obtain corresponding tight bounds for $\chi_{ON}(G)$ with respect to feedback vertex set and treewidth.

\bibliographystyle{splncs04}
\bibliography{main}
\end{document}